\documentclass[11pt,a4paper,notitlepage]{article}

\usepackage{amsmath}        
\usepackage{amsfonts}       
\usepackage{amsthm}         
\usepackage{bbding}         
\usepackage{bm}             
\usepackage{graphicx}       
\usepackage[square,sort,comma,numbers]{natbib}         
\usepackage[nottoc]{tocbibind} 
\usepackage{paralist}       
\usepackage[usenames]{xcolor}  

\usepackage{xfrac}
\usepackage{wrapfig}
\usepackage{caption}
\usepackage{titlesec}
\usepackage{filecontents}
\usepackage{mathtools}
\usepackage{mdwlist}
\usepackage{mathrsfs}
\usepackage{enumitem}
\usepackage{amssymb}
\usepackage[toc,page]{appendix}

\usepackage{pgfplots} 
\pgfplotsset{compat=1.7}
\usetikzlibrary{pgfplots.fillbetween}

\usepackage{tikz} 
\usetikzlibrary{patterns}
\usetikzlibrary{intersections}
\usetikzlibrary{positioning}

\usepackage[utf8]{inputenc} 
\usepackage{authblk} 
\usepackage{float} 

\usepackage[
cal=boondox,
calscaled=.94,
bb=ams,
frak=mma,
frakscaled=.97,
scr=boondox
]{mathalfa} 

\numberwithin{equation}{section}
\theoremstyle{definition}
\newtheorem{definition}{Definition}[section]

\newtheorem{lm}{Lemma}[section]

\title{Probabilistic Spacetimes}
\author[1]{Jakub K\'aninsk\'y}

\affil[1]{Institute of Theoretical Physics, Charles University in Prague}
\affil[  ]{jakubkaninsky@seznam.cz}

\date{\today}
\titleformat{\section}{\normalfont\scshape\LARGE}{\thesection}{1em}{}
\titleformat{\subsection}{\normalfont\LARGE}{\thesubsection}{1em}{}
\titleformat{\subsubsection}{\itshape\Large}{\thesubsubsection}{1em}{}

\renewenvironment{abstract}
 {\small
  \begin{center}
  \end{center}
  \list{}{
    \setlength{\leftmargin}{15mm}%
    \setlength{\rightmargin}{\leftmargin}%
  }
  \item\relax}
 {\endlist}

\begin{document}

\maketitle
\begin{abstract}
Probabilistic Spacetime is a~simple generalization of the classical model of spacetime in General Relativity, such that it allows to consider multiple metric field realizations endowed with probabilities. The motivation for such a generalization is a possible application in the context of some quantum gravity approaches, particularly those using the path integral. It is argued that this model might be used to describe simplified geometry, resulting e.g. from discretization, while keeping the continuous manifold; or it may be used as an effective description of a probabilistic geometry resulting from a full-fledged quantum gravity computation.
\end{abstract}

\vspace{4 mm}

\tableofcontents

\section{Preliminaries}

\subsection{Motivation}

The concept that we propose is motivated by quantum gravity, especially by the non-perturbative covariant approaches using sum over histories. These approaches take advantage of the well known Feynman's path integral formulation of quantum mechanics, which is (in case of non-relativistic quantum mechanics) equivalent to Schrodinger's equation, and constitutes a very general prescription for quantization. This prescription---as is believed---can be well applied to the case of gravity. In its essence, the path integral formulation \cite{Feynman2010} states that the probability $ P(b,a) $ of a system evolving from state $ a $ to state $ b $ is given as a square of an amplitude $ K(b,a) $, which is a sum of contributions from all paths between the states $ a $ and $ b $,
\begin{equation}
K(b,a) = \sum_{\text{paths}} e^{(i/\hbar)S[x(t)]},
\end{equation}
where $ S[x(t)] $ is the classical action associated with the path $ x(t) $. In case of quantum mechanics, the path is the position of a particle $ x $ as a function of the absolute time $ t $, and the sum over paths (or ``path integral'') can be defined rigorously using discretization of this absolute time and refining it in a limiting procedure.\\

In principle, one can use an analogical path integral for the gravitational field; however, it is not easy to provide a reliable definition, nor is it easy to extract a result once such definition is suggested. The sum-over-histories approaches differ widely in what they consider to be the \emph{state} and the \emph{path}. Furthermore, to provide the definition, discretization is usually needed, and it can either play a merely technical role (in which case the model must be eventually refined by a limiting procedure), or it can be interpreted as fundamental, in which case the spacetime is postulated to be discrete.\\

The sum over histories is a very strong principle, which allows us to approach quantum gravity in a manifestly covariant and conceptually clear way. To use this tool, one is bound to consider the set of all histories, and incorporate it into the theory in a way that allows for summing over it. These considerations raise questions about the fundamental structure of an element of this set, i.e., a spacetime. Each distinct way to treat the spacetime will produce a distinct theory, and it is of great interest to come up with such a treatment that would in the end produce a theory of the most favourable qualities (small amount of free parameters, computational accessibility of results, ...).\\

In order to illustrate different pathways towards addressing the problem, we will briefly compare a couple of covariant approaches connected with sum over histories.

\subsubsection*{Causal Sets}
Causal set is a discrete substratum for spacetime, which has been proposed by R. D. Sorkin in 1987 \cite{Sorkin1991} in an attempt to avoid some of the most common problems of quantum gravity at the time: nonrenormalizability within the standard canonical approach, the problem of singularity, and the computation of black hole entropy. All these have some kind of infinity involved, and it is argued that the solution might be found via the hypothesis that, at a subplanckian scale, the continuous manifold of General Relativity gives way to a ``discrete manifold'', carrying only a finite number of degrees of freedom in any finite volume. A nice feature of the discrete structure is that it may carry its metric relations within itself.\\

Causal set is a locally finite partially ordered set, or a set of elements with the relation $ \prec $ ``precedes'', which is transitive (if $ x \prec y $ and $ y \prec z $, then $ x \prec z $), acyclic (if $ x \prec y $ and $ y \prec x $, then $ x = y $) and locally finite (the number of elements $ y $ such that $ x \prec y \prec z $ is finite for all $ x, z $).\\

By specifying the ordering, the causal structure of the model is determined. The freedom in conformal transformation is then fixed by counting elements contained within some elementary volume and relating their number to the volume simply by $ N = V $, and thereby in some sense determining the volume element $ \sqrt{\vert g \vert} d x $. This blending of order with discreteness is what makes causal sets appealing. The single relation $ \prec $ codes (from macroscopic point of view) the topology, the differential structure, the geometry.\\

The model is very safe from diverging quantities, since all the computations are limited to, in one way or another, counting events. Thus, one can hope to provide a well defined path integral for the model, which would serve as a bridge to quantization. However, the problem is that it is generally very hard to relate the causal set to a spacetime manifold, since they are of very different nature (besides, one can of course recover the manifold only with limited precision). This separation of causal sets from the classical models and the difficulty to speak about geometry are maybe the main downsides of the theory.

\subsubsection*{Causal Dynamical Triangulations}
This more recent, still advancing approach \cite{Ambjoern2013} uses triangulations as a means to discretize the spacetime manifold and define a path integral on the discrete structure, only to refine it in a limiting procedure. The discrete structure is thus not considered fundamental in any way, it is only a tool to implement the path integral.\\

It has been shown on the example of 2D quantum gravity that it is possible to construct a well-defined lattice regularization with a UV lattice cut-off, which can be eventually removed to obtain a continuum limit.The advantage of the lattice approach is that there is no need for coordinates, as one employs the Regge calculus, and computes the Einstein-Hilbert action directly from the discretized geometry. A convenient choice is to use triangulations, i.e. lattices consisting of $ d $-simplices, which are $ d $-dimensional generalizations of triangles (2-simplices). The interior is assumed to be flat and the geometry is specified by giving the lengths of the $ \frac{d(d+1)}{2} $ edges. Together with the information about how the $ (d-1) $-dimensional boundary simplices are identified pairwise, this construction defines a triangulated manifold, which carries all the usual geometric information.\\

The path integral over geometries becomes the sum over all triangulations, with weight computed from the Regge implementation of Einstein-Hilbert action. In doing so, it is useful to restrict the set of piecewise linear geometries to a specific subset, namely the triangulations whose edges have all the same length. This subset of geometries is constructed by gluing together equilateral simplicial building blocks in all possible ways, compatible with certain constraints (topology, boundary conditions). Consequently, the variation in geometry is linked to the mutual connectivity of the building blocks created by the gluing and not to variations in the link length, giving rise to the name \emph{Dynamical Triangulations}.\\

In order to get a reasonably behaved theory, one may want to exclude geometries whose spatial topology is not constant in time. It turns out that these changes are associated with causality violations, and in the Lorentzian signature (i.e. before the Wick rotation is performed to switch into Euclidean geometry, which is a common practice for treating the integral), a general principle can be formulated which excludes these ill-behaved geometries. Then, one can take as a domain of the path integral the set of all Lorentzian piecewise flat triangulations whose causal structure is well defined, and where in particular no changes of the spatial topology are allowed to occur. This gives rise to \emph{Causal Dynamical Triangulations}.\\

\vspace{1 cm}
The two examples illustrate the span between path integral approaches to quantum gravity. In both of them, and other theories alike, discrete structure is used to establish some measure on geometries. In such theories, the usual unease lies in recovering the continuous manifold, be it because the model simply does not have sufficient information in it, or because the continuum limit is analytically and numerically out of reach. That is why we ask whether it would be possible, at least in some cases, to translate from the discrete model back to the \emph{continuous} manifold, describing the lack of information about geometry by some uncertain, probabilistic metric. Then, the simplified geometries in different approaches could be studied and compared within the usual general-relativistic model.\\

Another reason to consider the probabilistic version of the spacetime is a more general one, and does not relate to its micro-structure. We argue that ``probabilistic spacetime'' is what should result from a quantum-gravitational computation after the probability amplitudes are squared. In this sense, the model that we study within this article can be relevant in relation to any quantum gravity approach. There is also a possibility of using it as an effective large-scale model with probabilities provided from small-scale quantum gravity computations. In any case, such a generalization is a step towards the quantum nature, however small, and it may raise questions that are very relevant for our efforts.\\

\subsection{Probabilistic Metric Spaces}

The mathematical support and inspiration for our concept is provided by the well-established theory of probabilistic metric spaces. Its history dates back to 1942, when K. Menger introduced a generalization of Fréchet's abstract metric space, called a statistical metric space \cite{Menger1942}. Since then, the theory has been worked out to great extent, mainly by the most recent contributions of B. Schweizer and A. Sklar, whose book on the subject \cite{Schweizer2011} is our key reference.\\

To explain the idea behind probabilistic metric spaces, we will give word to K. Menger and partly reproduce the citation of his paper ``Modern Geometry and the Theory of Relativity'' from \cite{Schweizer2011}:

\begin{quote}

Poincaré, in several of his famous essays on the philosophy of science, characterized the difference between mathematics and physics as follows: In mathematics, if the quantity $ A $ is equal to the quantity $ B $, and $ B $ is equal to $ C $, then $ A $ is equal to $ C $; that is, in modern terminology: mathematical equality is a transitive relation. But in the observable physical continuum ``equal'' means indistinguishable; and in this continuum, if $ A $ is equal to $ B $, and $ B $ is equal to $ C $, it by no means follows that $ A $ is equal to $ C $. In the terminology of psychologists Weber and Fechner, $ A $ may lie within the threshold of $ B $, and $ B $ within the threshold of $ C $, even though $ A $ does not lie within the threshold of $ C $. 

\end{quote}

According to Menger, there is hope that this ``intransitive physical relation of equality'' could be retained in mathematics, and he suggests to do so with the use of probability. He writes:

\begin{quote}

Elaboration of this idea leads to the concept of a space in which a distribution function rather than a definite number is associated with every pair of elements. The number associated with two points of a metric space is called the distance between the two points. The distribution function associated with two elements of a statistical metric space might be said to give, for every $ x $, the probability that the distance between the two points in question does not exceed $ x $. Such a statistical generalization of metric spaces appears to be well adapted for the investigation of physical quantities and physiological thresholds. The idealization of local behaviour of rods and boards, implied by this statistical approach, differs radically from that of Euclid. In spite of this fact, or perhaps just because of it, the statistical approach may provide a useful means for geometrizing the physics of the microcosm.

\end{quote}

The idea of imprecise measurement and a geometry that incorporates the imprecision is remarkably old; and remarkably well resonates with the idea that the spacetime may have a fundamental limitation on the resolution of points implied by its quantum nature. We study the mathematical concept inspired by this idea and attempt to use it in favour of the present efforts. In the latter, the fundamentals of the theory of probabilistic metric spaces will be introduced, serving partly as a rigorous base for our future work and partly as a guideline for the line of thought that we might want to develop.

\subsubsection*{Notation}

Throughout the whole text, the set of real numbers is denoted $ \mathbb{R} $ and the set of real numbers extended with symbols $ - \infty $ and $ + \infty $ is denoted $ R $. The unit interval $ [0,1] $ will be denoted $ I $.

\subsubsection*{Probabilities}

\begin{definition}   \label{def:sigma_field}
Let $ \Omega $ be a nonempty set. A sigma field on $ \Omega $ is a family $ \mathscr{P} $ of subsets of $ \Omega $ such that:
\begin{enumerate}[label=\roman*., itemsep=-0.2mm]
\item $ \Omega $ is in $ \mathscr{P} $.
\item If $ A $ is in $ \mathscr{P} $, then $ \Omega \setminus A $, the complement of $ A $, is in $ \mathscr{P} $.
\item If $ A_n $ is in $ \mathscr{P} $ for $ n = 1, 2, ... $, then $ \bigcup_{n = 1}^{\infty} A_n $ is in $ \mathscr{P} $.
\end{enumerate}
\end{definition}

\vspace{4 mm}

We remark that if $ \Omega $ is a closed interval, then there is a unique smallest sigma field on $ \Omega $ that contains all the subintervals of $ \Omega $; this is the \emph{Borel field} on $ \Omega $, and its members are \emph{Borel sets} in $ \Omega $.

\begin{definition}   \label{def:measurable_function}
Let $ \mathscr{P} $ be a sigma field on a set $ \Omega $. A function $ f $ from $ \Omega $ into $ R $ is \emph{measurable with respect to $ \mathscr{P} $} if, for every $ x $ in $ R $, the inverse image $ f^{-1}[ -\infty, x ) $ of the interval\footnote{In greater detail, this means: the set $ \lbrace \omega ~ \vert ~ \exists y \in [ -\infty, x ): f^{-1}(y) = \omega \rbrace $ is in $ \mathscr{P} $.} $ [ -\infty, x ) $ is in $ \mathscr{P} $. If $ f $ is measurable with respect to the Borel field on a closed interval, then $ f $ is \emph{Borel measurable}.
\end{definition}

Now we know which functions we can measure. The measure itself is constructed as follows.

\begin{definition}   \label{def:probability_space}
A \emph{probability space} is a triple $ (\Omega, \mathscr{P}, P) $, where $ \Omega $ is a \break nonempty set, $ \mathscr{P} $ is a sigma field on $ \Omega $, and $ P $ is a function from $ \mathscr{P} $ into $ I $ such that the following conditions hold:
\begin{enumerate}[label=\roman*., itemsep=-0.2mm]
\item $ P(\Omega) = 1 $ and $ P(\varnothing) = 0 $.
\item If $ \lbrace A_n \rbrace $ is a sequence of pairwise disjoint sets in $ \mathscr{P} $, then
\begin{equation}
P \left( \bigcup_{n = 1}^{\infty} A_n \right) = \sum_{n = 1}^{\infty} P(A_n).
\end{equation}
\end{enumerate}
The function $ P $ is a \emph{probability measure on $ \Omega $}.
\end{definition}

\vspace{4 mm}

If $ \Omega $ is a closed interval $ [a, d] $ and $ F $ is a nondecreasing function on $ \Omega $ with $ F(a) = 0 $ and $ F(d) = 1 $, then $ F $ defines a unique probability measure $ P_F $, called the \emph{Lebesque-Stieltjes $ F $-measure on $ \Omega $}.

Since a probability measure, like any function, determines its domain, it is always possible to speak of a function as being \emph{$ P $-measurable}, rather than measurable with respect to the sigma field $ \mathscr{P} = \mathrm{Dom}~P $.

Any probability measure $ P $ determines a corresponding integral. When it exists, the integral of a $ P $-measurable function $ f $ with respect to $ P $ over a set $ A $ in $ \mathrm{Dom}~P $ is denoted by $ \int_A f \mathrm{d}P $. When dealing with a Lebesque-Stieltjes measure $ P_F $, we generally write $ \int_A f \mathrm{d}F $ rather than $ \int_A f \mathrm{d}P_F $.

\begin{definition}   \label{def:random_variable}
A \emph{random variable} on a probability space $ (\Omega, \mathscr{P}, P) $ is a \\ $ P $-measurable function on $ \Omega $.
\end{definition}

\subsubsection*{Probabilistic Metric Spaces}

\begin{definition}   \label{def:df}
A \emph{distribution function} (briefly, a \emph{d.f.}) is a nondecreasing function $ F $ defined on $ R $, with $ F(-\infty) = 0 $ and $ F(\infty) = 1 $. The set of all distribution functions that are left continuous\footnote{An example of a distribution function which is left-continuous but not right-continuous is given by Def. \ref{def:epsilon}.} on $ \mathbb{R} $ will be denoted by $ \Delta $. The subset of $ \Delta $ consisting of those elements $ F $ such that $ \lim_{x \to - \infty} F(x) = 0 $ and $ \lim_{x \to + \infty} F(x) = 1 $ will be denoted by $ \mathscr{D} $.
\end{definition}

\begin{definition}   \label{def:df_random_variable}
Let $ X $ be a random variable on the probability space \break $ (\Omega, \mathscr{P}, P) $. Then $ F_X $ is the function on $ R $ defined by $ F_X(-\infty) = 0 $, $ F_X(\infty) = 1 $, and
\begin{equation}
F_X(x) = P \lbrace \omega \text{ in } \Omega \vert ~ X(\omega) < x \rbrace ~~ \text{ for } -\infty < x < \infty.
\end{equation}
\end{definition}

\begin{definition}   \label{def:pdf}
Let $ X $ be a random variable on the probability space \break $ (\Omega, \mathscr{P}, P) $. A \emph{probability density function}\footnote{In \cite{Schweizer2011}, there is no notion of a p.d.f. We are adding this definition for our convenience.} (briefly, a \emph{p.d.f.}) is a function $ f_X $ from $ R $ into $ [0, \infty) $, such that
\begin{equation}
\int_a^{b} f_X(x) dx = P \lbrace \omega \text{ in } \Omega \vert ~ a \leq X(\omega) < b \rbrace \qquad \text{ for } a,b \in \mathbb{R}.
\end{equation}
\end{definition}

\begin{definition}
A \emph{distance distribution function} (briefly, a \emph{d.d.f.}) is a nondecreasing function $ F $ defined on $ R^{+} $ that satisfies $ F(0) = 0 $ and $ F(\infty) = 1 $ and is left-continuous on $ \mathbb{R}^+ $. The set of all distance distribution functions will be denoted by $ \Delta^{+} $ and the set of all $ F $ in $ \Delta^{+} $ for which $ \lim_{x \to + \infty} F(x) = 1 $ by $ \mathscr{D}^{+} $. Distance distribution functions are ordered by defining $ F \leq G $ to mean $ F(x) \leq G(x) $ for all $ x > 0 $.
\end{definition}

\begin{definition} \label{def:epsilon}
For any $ a $ in $ \mathbb{R} $, $ \varepsilon_a $, the \emph{unit step} at $ a $, is the function in $ \Delta $ given~by
\begin{equation*}
\varepsilon_a(x) = 
  \begin{cases}
    0,       & \quad -\infty \leq x \leq a,  \\
    1,  & \quad ~ ~ ~ ~ a < x \leq \infty. \\
  \end{cases}
\end{equation*}
\end{definition}

\begin{definition}
A \emph{binary operation} on a nonempty set $ S $ is a function $ T $ from $ S \times S $ into $ S $, i.e., a function $ T $ with $ \mathrm{Dom} ~ T = S \times S $ and $ \mathrm{Ran} ~ T \subseteq S $.
\end{definition}

\begin{definition}
A \emph{triangular norm} (briefly, a \emph{t-norm}) is an associative binary operation on $ I $ that is commutative, nondecreasing in each place, and such that $ T(a, 1) = a $ for all $ a $ in $ I $.
\end{definition}

\begin{definition} \label{def:triangle_function}
A \emph{triangle function} is a binary operation on $ \Delta^{+} $ that is commutative, associative, nondecreasing in each place and has $ \varepsilon_0 $ as identity.
\end{definition}

In order to classify some of the most important triangle functions, we first denote $ \mathscr{T} $ the set of binary operations on $ I $ that have $ 1 $ as identity; and the class $ \mathscr{L} $ to be the set of all binary operations $ L $ on $ \mathbb{R}^{+} $ that are nondecreasing, continuous on $ \mathbb{R}^{+} \times \mathbb{R}^{+} $ and have $ \mathrm{Ran}~L = \mathbb{R}^{+} $. Now we can write:

\begin{definition}
For any $ T $ in $ \mathscr{T} $ and $ L $ in $ \mathscr{L} $, $ \tau_{T,L} $ is the function on $ \Delta^{+} \times \Delta^{+} $ whose value, for any $ F, G $ in $ \Delta^{+} $, is the function $ \tau_{T,L}(F,G) $ defined on $ \mathbb{R}^+ $ by
\begin{equation*}
\tau_{T,L}(F,G)(x) = \sup_{\lbrace (u,v) \rbrace} \lbrace T(F(u), G(v)) \vert L(u,v) = x \rbrace.
\end{equation*}
If $ L = \mathrm{Sum} $, then we simply write $ \tau_T $.
\end{definition}

\begin{definition} \label{def:PMspace}
A \emph{probabilistic metric space} (briefly, a PM space) is a triple $ (S, \mathscr{F}, \tau) $ where $ S $ is a nonempty set (whose elements are the \emph{points} of the space), $ \mathscr{F} $ is a function from $ S \times S $ into $ \Delta^{+} $, $ \tau $ is a triangle function, and the following conditions are satisfied for all $ p, q, r $ in $ S $:
\begin{enumerate}[label=\roman*., itemsep=-0.2mm]
\item $ \mathscr{F}(p, p) = \varepsilon_0 $.
\item $ \mathscr{F}(p, q) \neq \varepsilon_0 $ if $ p \neq q $.
\item $ \mathscr{F}(p, q) = \mathscr{F}(q, p) $.
\item $ \mathscr{F}(p, r) \geq \tau(\mathscr{F}(p, q),\mathscr{F}(q, r)) $.
\end{enumerate}
\end{definition}

Probabilistic metric spaces can be subjected to a partial classification based on the properties of the triangle function $ \tau $, as follows.

\begin{definition} \label{def:triangle_function_classification}
Let $ (S, \mathscr{F}, \tau) $ be a probabilistic metric space. Then $ (S, \mathscr{F}, \tau) $ is proper if
\begin{equation*}
\tau(\varepsilon_a, \varepsilon_b) \geq \varepsilon_{a + b} ~ ~ \mathrm{for ~ all} ~ a,b ~ \mathrm{in} ~ \mathbb{R}^{+}.
\end{equation*}
If $ \tau = \tau_T $ for some t-norm $ T $, then $ (S, \mathscr{F}, \tau) $ is a \emph{Menger space}, or equivalently, $ (S, \mathscr{F}) $ is a \emph{Menger space under} $ T $. If $ \tau $ is convolution, then $ (S, \mathscr{F}) $ is a \emph{Wald space}.
\end{definition}

\subsubsection*{Distribution-Generated Spaces}
Now we will turn to another kind of spaces which is very closely related to probabilistic metric spaces and also very inspirational for what we will try to achieve later. They are called \emph{distribution-generated spaces}.\\

First, we have to be able to work with multidimensional distribution functions.

\begin{definition}
Let $ n $ be a positive integer. An $ n $-interval is the Cartesian product of $ n $ real intervals, and an $ n $-box the Cartesian product of $ n $ closed intervals. If $ J $ is the $ n $-interval $ J_1 \times J_2 \times ... \times J_n $, where for $ m = 1, 2, ..., n $ the interval $ J_m $ has endpoints $ a_m, e_m $, then the vertices of $ J_m $ are the points $ c = (c_1, ..., c_n) $ such that each $ c_m $ is equal to either $ a_m $ or $ e_m $.
\end{definition}

\begin{definition}
Let $ B $ be the $ n $-box $ [a, e] $, where $ a = (a_1, ..., a_n) $ and $ e = (e_1, ..., e_n) $. Let $ c = (c_1, ..., c_n) $ be a vertex of $ B $. If the vertices of $ B $ are all distinct (which is equivalent to saying $ a < e $), then
\begin{equation}
\mathrm{sgn}_B (c) =   \begin{cases}
    \phantom{-}1  & \quad \text{if } c_m = a_m \text{ for an even number of } m\text{'s},\\
    -1 & \quad \text{if } c_m = a_m \text{ for an odd number of } m\text{'s}.\\
  \end{cases}
\end{equation}
If the vertices of $ B $ are not all distinct, then $ \mathrm{sgn}_B (c) = 0 $.
\end{definition}

An illustration for this definition is provided by Fig. \ref{fig:volume}. Using the $ \mathrm{sgn}_B $, we can generalize the classical Newton-Leibniz formula to higher dimensions, as is done by the following definition.

\begin{figure}[ht]
\centering
\begin{tikzpicture}[scale=1]
	 \tikzstyle{vertexA}=[circle,minimum size=7pt,inner sep=0pt,fill=black!30]
	 \tikzstyle{vertexB}=[circle,minimum size=7pt,inner sep=0pt,fill=black]
	 \tikzstyle{selected vertex} = [vertex, fill=red!24]
	 \tikzstyle{selected edge} = [draw,line width=5pt,-,red!50]
	 \tikzstyle{edge} = [draw,-,gray]
	 \node[vertexA] (v0) at (0,0) {};
	 \node[vertexB] (v1) at (0,1) {};
	 \node[vertexB] (v2) at (1,0) {};
	 \node[vertexA] (v3) at (1,1) {};
	 \node[vertexB] (v4) at (0.23, 0.4) {};
	 \node[vertexA] (v5) at (0.23,1.4) {};
	 \node[vertexA] (v6) at (1.23,0.4) {};
	 \node[vertexB] (v7) at (1.23,1.4) {};
	 \node[vertexB] (v8) at (-1,-1) {};
	 \node[vertexA] (v9) at (-1,2) {};
	 \node[vertexB] (v13) at (-0.66,2.7) {};
	 \node[vertexA] (v12) at (-0.66,-0.3) {};
	 \node[vertexA] (v10) at (2,-1) {};
	 \node[vertexB] (v14) at (2.34,-0.3) {};
	 \node[vertexB] (v11) at (2,2) {};
	 \node[vertexA] (v15) at (2.34,2.7) {};
	 \draw[edge] (v0) -- (v1) -- (v3) -- (v2) -- (v0);
	 \draw[edge] (v0) -- (v4) -- (v5) -- (v1) -- (v0);
	 \draw[edge] (v2) -- (v6) -- (v7) -- (v3) -- (v2);
	 \draw[edge] (v4) -- (v6) -- (v7) -- (v5) -- (v4);
	 \draw[edge] (v8) -- (v9) -- (v13) -- (v12) -- (v8);
	 \draw[edge] (v0) -- (v4) -- (v12) -- (v8) -- (v0);
	 \draw[edge] (v1) -- (v9) -- (v13) -- (v5) -- (v1);
	 \draw[edge] (v2) -- (v10) -- (v14) -- (v6) -- (v2);
	 \draw[edge] (v8) -- (v10) -- (v14) -- (v12) -- (v8);
	 \draw[edge] (v3) -- (v11) -- (v15) -- (v7) -- (v3);
	 \draw[edge] (v10) -- (v11) -- (v15) -- (v14) -- (v10);
	 \draw[edge] (v9) -- (v11) -- (v15) -- (v13) -- (v9);
	 \draw[edge] (v0) -- (v2);
	 \draw[edge] (v2) -- (v6);
	 \draw[edge] (v6) -- (v4);
	 \draw[edge] (v4) -- (v5);
	 \draw[edge] (v5) -- (v13);
	 \draw[edge] (v13) -- (v12);
	 \draw[edge] (v12) -- (v14);
	 \draw[edge] (v14) -- (v15);
	 \draw[edge] (v15) -- (v7);
	 \draw[edge] (v7) -- (v3);
	 \draw[edge] (v3) -- (v1);
	 \draw[edge] (v1) -- (v9);
	 \draw[edge] (v9) -- (v11);
	 \draw[edge] (v11) -- (v10);
	 \draw[edge] (v10) -- (v8);
	 \draw[edge] (v8) -- (v0);
 \end{tikzpicture}
\vspace{-2 mm}
\caption{Diagram of a 4-box with vertices coloured gray or black, distinguishing between the two possible values of $ \mathrm{sgn}_B $.}
\label{fig:volume}
\end{figure}
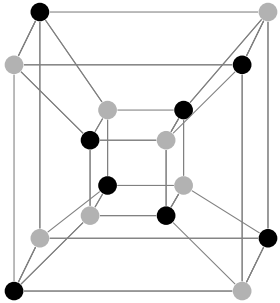

\begin{definition}
Let $ A_1, ..., A_n $ be nonempty subsets of $ R $, and let $ H $ be a function from $  A_1 \times ... \times A_n  $ into $ \mathbb{R} $. If $ B $ is an $ n $-box whose vertices are all in $ \mathrm{Dom}~H $, then the \emph{$ H $-volume of $ B $} is the sum
\begin{equation}
V_H(B) = \sum_{\mathrm{vertices ~ of ~ } B} \mathrm{sgn}_B (c) H(c).
\end{equation}
The function $ H $ is $ n $-increasing if $ V_H(B) \geq 0 $ for all $ n $-boxes $ B $ whose vertices are in $ \mathrm{Dom}~H $.
\end{definition}

\begin{definition}
Let $ H $ be a function from $ A_1 \times ... \times A_n $ into $ \mathbb{R} $, where each $ A_m $ is a subset of $ R $ that contains a least element $ a_m $. Then $ H $ is \emph{grounded} if $ H(x_1, ..., x_n) = 0 $ for all $ (x_1, ..., x_n) $ in $ \mathrm{Dom}~H $ such that $ x_m = a_m $ for at least one $ m $.
\end{definition}

\begin{definition}
Let $ n $ be an integer $ \geq 2 $. An \emph{$ n $-dimensional distribution function} (briefly, an \emph{$ n $-d.f.}) is a function $ H $ satisfying the following conditions:
\begin{enumerate}[label=\roman*.]
\item $ \mathrm{Dom}~H = R^{n} $,
\item $ H $ is $ n $-increasing and grounded,
\item $ H(\infty, \infty, ..., \infty) = 1 $.
\end{enumerate}
A \emph{joint distribution function} (briefly, a \emph{joint d.f.}) is a function $ H $ for which there is an integer $ n \geq 2 $ such that $ H $ is an $ n $-d.f.
\end{definition}

\begin{definition}
Let $ H $ be a function from $ A_1 \times ... \times A_n $ into $ \mathbb{R} $. Suppose each $ A_m $ is a nonempty subset of $ R $ and has a maximal element $ e_m $. Then $ H $ \emph{has margins}. The \emph{one-dimensional margins} of $ H $ are the functions $ H_m $ given by
\begin{equation}
\begin{gathered}
\mathrm{Dom} ~ H_m = A_m \qquad \mathrm{for} ~ m = 1,2, ..., n,\\
H_m(x) = H(e_1, ..., e_{m-1}, x, e_{m+1}, ..., e_n)
\end{gathered}
\end{equation}
for all $ x $ in $ \mathrm{Dom} ~ H_m $. The higher-dimensional margins are defined similarly, i.e., by fixing fewer places in $ H $.
\end{definition}

It can be shown that every one-dimensional margin of a joint d.f. is a distribution function. We shall denote those one-dimensional margins of $ H $ by $ F_1, ..., F_n $ and refer to them briefly as ``margins''.

\begin{definition}
Let $ n $ be an integer $ \geq 2 $. An \emph{$ n $-dimensional subcopula} (briefly, an \emph{$ n $-subcopula}) is a function $ C' $ that satisfies the following conditions:
\begin{enumerate}[label=\roman*.]
\item $ \mathrm{Dom} ~ C' = A_1 \times ... \times A_n $ where each $ A_m $ is a subset of $ I $ containing $ 0 $ and $ 1 $,
\item $ C' $ is $ n $-increasing and grounded,
\item For every positive integer $ m \leq n $, the margin $ C_m' $ of $ C' $ satisfies
\begin{equation}
C_m'(x) = x \qquad \text{for all } x \text{ in } A_m.
\end{equation}
\end{enumerate}
An \emph{$ n $-dimensional copula} (briefly, an \emph{$ n $-copula}) is an $ n $-subcopula $ C $ whose domain is the entire unit $ n $-cube $ I^{n} $.
\end{definition}

Just as a function $ F $ in $ \Delta $ defines a Lebesque-Stieltjes probability measure\footnote{For details, see Sections 2.3 and 6.5 of \cite{Schweizer2011} and references therein.} $ P_F $ on $ R $, so an $ n $-d.f. $ H $ defines a Lebesque-Stieltjes probability measure $ P_H $ on $ R^{n} $, and an $ n $-copula $ C $ defines a Lebesque-Stieltjes probability measure $ P_C $ on $ I^{n} $. If $ B $ is an $ n $-box in $ I^{n} $ and $ C $ is an $ n $-copula, then
\begin{equation}
\int_B \mathrm{d}C = P_C(B) = V_C(B).
\end{equation}
Similarly, if $ H $ is an $ n $-d.f. whose margins are in $ \Delta $, then
\begin{equation}
\int_{[(-\infty, ..., -\infty), u)} \mathrm{d}H = P_H [(-\infty, ..., -\infty), u) = H(u), \label{int_PH}
\end{equation}
where $ u $ is an $ n $-tuple from $ \mathrm{Dom} ~ H $ and $ (-\infty, ..., -\infty) $ is the $ n $-tuple of infinities. If the margins of $ H $ are only in $ \Delta^{+} $, then
\begin{equation}
\int_{[0, u)} \mathrm{d}H = P_H [0, u) = H(u).
\end{equation}

Now we can approach the definition of the distribution-generated space.\\

Let $ S $ be a set. With each point $ p $ of $ S $ associate an $ n $-dimensional distribution function $ G_p $ whose margins are in $ \mathscr{D} $, and with each pair $ (p, q) $ of distinct points of $ S $ associate a $ 2n $-dimensional distribution function $ H_{pq} $ such that
\begin{equation}
H_{pq}(u, (\infty, ...,\infty)) = G_p(u), \qquad H_{pq}((\infty, ...,\infty), v) = G_q(v) \label{HG}
\end{equation}
for any $ u = (u_1, ..., u_n) $ and $ v = (v_1, ..., v_n) $ in $ R^n $. For any $ x \geq 0 $, let $ Z(x) $ be the cylinder in $ R^{2n} $ given by
\begin{equation}
Z(x) = \lbrace (u,v) ~ \mathrm{in} ~ \mathbb{R}^{2n} \vert ~ \vert u - v \vert < x \rbrace
\end{equation}
and define $ F_{pq} $ in $ \mathscr{D}^+ $ via
\begin{equation}
F_{pq}(x) = \int_{Z(x)} \mathrm{d}H_{pq}. \label{F_pq}
\end{equation}

In order to obtain a PM space by means of the foregoing construction, the definition of $ F_{pg} $ via (\ref{F_pq}) must be extended to the case in which $ p = q $ in such a way that (\ref{HG}) holds and $ F_{pp} = \varepsilon_0 $. We achieve this with the aid of
\begin{lm} \label{lem:dgs}
Let $ G $ be an $ n $-dimensional distribution function all of whose margins are in $ \mathscr{D} $, and let $ H $ be the $ 2n $-dimensional distribution function given by
\begin{equation}
H(u,v) = G(\mathrm{Min}(u_1, v_1), ..., \mathrm{Min}(u_n, v_n)) \label{Huv}
\end{equation}
for all $ u = (u_1, ..., u_n) $ and $ v = (v_1, ..., v_n) $ in $ R^{n} $. Then
\begin{enumerate}[label=\roman*.]
\item $ H(u, (\infty, ..., \infty)) = H((\infty, ..., \infty), u) = G(u) $ for all $ u $ in $ R^{n} $.
\item For any positive integer $ m \leq n $ is $ H(u_1, ..., u_m, ..., u_n, v_1, ... , v_m, ..., v_n) = $\\ $ H(u_1, ..., v_m, ..., u_n, v_1, ... , u_m, ..., v_n) $.
\item $ P_H(Z(x)) = 1 $ for all $ x > 0 $.
\end{enumerate}
\end{lm}
\begin{proof}
We will omit the proof.\footnote{It can be found in \cite{Schweizer2011}, Chapter 10, p. 158.}
\end{proof}

It follows from Lemma \ref{lem:dgs} that if $ H_{pp} $ is defined by (\ref{Huv}) with $ G = G_p $, and $ F_{pp} $ by (\ref{F_pq}), then $ F_{pp} = \varepsilon_0 $. Furthermore, if $ H_{pq}(u,v) = H_{qp}(v,u) $ then $ F_{pq} = F_{qp} $, which brings us to

\begin{definition} \label{def:dgs}
Let $ S $ be a set, $ \mathscr{F} $ a function from $ S \times S $ into $ \Delta^{+} $, and $ n $ a positive integer. Then $ (S, \mathscr{F} ) $ is a distribution-generated space (over $ R^{n} $) if the following conditions are satisfied:
\begin{enumerate}[label=\roman*.]
\item For every $ p $ in $ S $ there is an $ n $-d.f. $ G_p $ whose margins $ G_{p 1}, ..., G_{p n} $ are in $ \mathscr{D} $.
\item For every $ p,q $ in $ S $ there is a $ n $-d.f. $ H_{pq} $ which satisfies (\ref{HG}), is such that $ H_{pq}(u,v) = H_{qp}(v,u) $ for all $ u,v $ in $ R^{n} $, and is given by (\ref{Huv}) with $ G = G_p $ when $ p = q $.
\item For every $ p, q $ in $ S $, $ F_{pq} $ is given by (\ref{F_pq}).
\end{enumerate}
\end{definition}

\section{Probabilistic Spacetime}

\subsection{Introducing Probability into Geometry}

We have reached the point in which we want to define the ``probabilistic spacetime''. Roughly speaking, it should be a collection of metric tensor fields defined on a manifold, endowed with probabilities. Indeed, we will be considering \emph{a single manifold}, that is, the ``probabilistic spacetime'' in our treatment will have a fixed topology.\\

If we want a truly probabilistic object, i.e., an object for which the probabilities of distinct outcomes fulfil the usual rules, we should seek inspiration in probability theory. After trying out numerous variants, our efforts have resulted in the following implementation. It starts from the most general notions and continues towards more specific concepts, allowing for considerations on all levels of abstraction.\\

As it turns out, the very first object that we need is at the same time the most interesting one:

\begin{definition}
\emph{The space of metric tensor realizations} $ \mathscr{R} $ on the manifold $ M $ is the set of all Lorentzian metric tensor fields $ g $ on $ M $.
\end{definition}

Note that we are considering \emph{all} possible metrics, not only equivalence classes given by diffeomorphisms. Manifold is a \emph{set} with an atlas, and while the atlas is said to be maximal and comprises all mutually compatible choices of maps (which particularly means all different coordinate charts), the elements of the manifold (i.e., distinct events) are unique. Thus, there are non-equivalent elements $ g_A, g_B $ in $ \mathscr{R} $ such that there is a diffeomorphism $ \phi $ on $ M $ inducing $ g_B $ from $ g_A $. This is important since we need to have the information about where in the manifold we are.\\

We can build a probabilistic space on $ \mathscr{R} $:

\begin{definition}
\emph{Probabilistic metric tensor field} $ \mathcal{g} $ is a random variable on the probability space $ (\mathscr{R}, \mathscr{P}^{g}, P^{g}) $.
\end{definition}

If we want to work with $ \mathcal{g} $, we first have to define some convenient, useful sigma field $ \mathscr{P}^{g} $ and, subsequently, the measure $ P^{g} $. We should be careful though. The space of metric tensor realizations $ \mathscr{R} $ is $ \infty $-dimensional so it is not easy to define a ``fine'' measure on it. For instance, we cannot use the Borel measure. That is unfortunate, since it does not give us much hope for defining a measure on the ``perfect'' sigma field:

\begin{definition}
\emph{Perfect sigma field} on $ \mathscr{R} $ is the set $ \mathscr{P}^{g}_{\text{perf.}} = \lbrace X \vert X \subseteq \mathscr{R} \rbrace $.
\end{definition}

One can see that such a set is indeed a sigma field. We will later show that it is possible to have a useful measure on $ \mathscr{P}^{g}_{\text{perf.}} $, albeit a trivial one.\\

Eventually, the awaited definition:

\begin{definition} \label{def:pst}
\emph{Probabilistic spacetime} (briefly, \emph{p.st.}) $ \mathcal{M} = (M, \mathcal{g}) $ is an $ n $-dimensional, $ C^{\infty} $, real manifold $ M $ together with the probabilistic metric tensor field $ \mathcal{g} $.
\end{definition}

Since any probability measure determines a corresponding integral, we now can integrate a $ P^{g} $-measurable function $ f $ w.r.t. $  P^{g} $ over a set $ A $ in $ \mathrm{Dom} ~ P^{g} $. The integral is denoted  $ \int_A f ~ d P^{g} $. In applications, if there is no space for confusion, we will write simply $ d G $ instead of $ d P^{g} $.\\

As we have already noted, the theory is to no use if we cannot provide a useful sigma algebra with some measure on $ \mathscr{R}  $. In the following, we show maybe the easiest way to do that.

\subsubsection*{Fragmentation of $ \mathscr{R} $}

\begin{definition}
\emph{Fragmentation $ y $ of $ \mathscr{R} $} is a function from $ \mathscr{R} $ onto a set $ S $. If $ S $ is countable, we call the fragmentation \emph{discrete}, if $ S $ is homomorphic to $ \mathbb{R}^{N} $, we call the fragmentation \emph{$ N $-parametric}.
\end{definition}

\begin{definition}
Let there be a fragmentation $ y $ from $ \mathscr{R} $ onto a set $ S $. Then the preimage $ y^{-1}(s) $ of an element $ s $ of $ S $ is called a \emph{fragment of $ \mathscr{R} $}.
\end{definition}

\begin{lm} \label{lm:frag}
Let there be a fragmentation $ y $ from $ \mathscr{R} $ onto a set $ S $. Then a sigma field $ \mathscr{P}^{S} $ on $ S $ uniquely defines a sigma field $ \mathscr{P}^{g} $ on $ \mathscr{R} $.
\end{lm}
\begin{proof}
Let $ \mathscr{P}^{S} $ be a sigma field on $ S $ and $ \mathscr{P}^{g} = \lbrace y^{-1}(Q) \vert \text{ for all } Q \text{ from } \mathscr{P}^{S} \rbrace $. We need to check the following:
\begin{enumerate}[label=\roman*., itemsep=-0.2mm]
\item By definition, there is $ S $ in $ \mathscr{P}^{S} $. Then there is a set $ y^{-1}(S) $ in $ \mathscr{P}^{g} $. But $ y^{-1}(S) = \mathscr{R} $ (because the fragmentation is surjective, $ y(g) $ is always in $ S $), so $ \mathscr{R} $ is in $ \mathscr{P}^{g} $.
\item For a set $ A $ in $ \mathscr{P}^{S} $, the set $ S \smallsetminus A $ is also in $ \mathscr{P}^{S} $. Thus, for a set $ B $ in $ \mathscr{P}^{g} $, such that $ B = y^{-1}(A) $, there is also a set $ y^{-1}(S \smallsetminus A) $. It is equal to $ \mathscr{R} \smallsetminus B $.
\item Analogical argument can be easily done for a unification $ \bigcup_{n = 1}^{\infty} B_n $ of sets $ B_n $ from $ \mathscr{P}^{g} $.
\end{enumerate}
Thus, according to Def. \ref{def:sigma_field}, $ \mathscr{P}^{g} $ is a sigma field.
\end{proof}

Lemma \ref{lm:frag} is indeed useful. Although we cannot define the Borel sigma field on $ \mathscr{R} $, we can define it on the $ N $-dimensional set $ S $ and it will determine some ``rough'' sigma field on $ \mathscr{R} $. Then we can easily provide it with measure by copying a chosen measure on $ S $. This construction enables us to measure probabilities of sets of fragments of $ \mathscr{R} $.\\

The discrete fragmentation is not much interesting, and that is for two reasons. First, the sigma field $ \mathscr{P}^{g} $ arising from a discrete fragmentation is arguably too simple to be of any use; and second, we can introduce such a sigma field straightforwardly without using the function $ y $. Still, the definitions provide an intuitive description, making it easy to speak about the simple sigma fields with discrete or Borel measures.\\

An example of an $ N $-parametric fragmentation will be given in the next paragraph.

\subsubsection*{Locally Probabilistic Spacetime}

In some special cases, where we do not need to speak about the metric tensor \emph{field}, but rather only the metric tensors in distinct points of the manifold, it suffices to utilize a much simpler concept. It is built analogically, but with a major difference.

\begin{definition}
\emph{Probabilistic tensor} $ \mathcal{T} $ is a random variable on the probability space $ (T^{(r,s)}, \mathscr{P}^T, P^T) $, where $ T^{(r,s)} $ is $ r $-times contravariant and $ s $-times covariant tensor space; $ \mathscr{P}^T $ is the sigma field on $ T^{(r,s)} $; and $ P^T $ is the probability measure on $ T^{(r,s)} $.
\end{definition}

Unlike in the preceding case, there is a Borel field on $ T^{(r,s)} $. If $ T^{(r,s)} $ is to be a tensor space in some point of an $ n $-dimensional manifold, then it is homomorphic to $ \mathbb{R}^{m} $, where $ m = (r+s)^{n} $. Further, it suffices to provide an $ m $-distribution function $ H $ on $ T^{(r,s)} $ and it will uniquely define the Lebesque-Stieltjes integral on $ T^{(r,s)} $, as noted by (\ref{int_PH}). The random variable $ \mathcal{T} $ is uniquely defined via $ H $. We can also use the $ m $-dimensional p.d.f. $ h(x) $, where $ x $ stands for the $ m $-tuple $ (x_1, ..., x_m) $, satisfying
\begin{equation}
\int_{(-\infty,...,-\infty)}^{x} h(x') ~ dx' = H(x).
\end{equation}
The Lebesque-Stieltjes integral defined by $ P^T_H $ will be denoted $ \int dH $. The probability that a \emph{realization} $ T $ of $ \mathcal{T} $ is from the set $ Q \in \mathscr{P}^T $ (in other words, $ Q \subseteq T^{(r,s)} $), is $ \int_Q dH $. One typically computes this integral by converting it into the Riemann integral by replacing $ dH $ by $ h(x) ~ dx = h(x) ~ dx_1 ... dx_m  $.\\

Now that we have probabilistic tensors, we can treat the metric tensor easily:

\begin{definition} \label{def:prob_metric_tensor}
\emph{Probabilistic metric tensor} $ \mathcal{g}_x $ is a probabilistic tensor on the probability space $ (T^{g(x)}, \mathscr{P}^{g(x)}, P^{g(x)}) $, where $ T^{g(x)} $ is the tangent space $ T^{(0,2)}_x M $ in a point $ x $ of $ M $ narrowed to symmetric, nondegenerate Lorentzian tensors.
\end{definition}

\begin{definition} \label{def:lpst}
\emph{Locally probabilistic spacetime} (briefly, \emph{l.p.st.}) $ \mathcal{M}_{\text{loc.}} = (M, \mathcal{g}_x) $ is an $ n $-dimensional, $ C^{\infty} $, real manifold $ M $ together with the probabilistic metric tensor $ \mathcal{g}_x $ defined in every point $ x $ of $ M $.
\end{definition}

The probabilistic metric tensor $ \mathcal{g}_x $ in a point $ x $ determines a roughly-measured p.st. via an $ N $-parametric fragmentation (for a 4D spacetime, $ N = 10 $). Let $ y $ be a function from $ \mathscr{R} $ onto $ T^{g(x)} $ (Def. \ref{def:prob_metric_tensor}) for some point $ x $ of $ M $, such that $ y(g) = g(x) $. Then we can define $ \mathscr{P}^{g} $ and $ P^{g} $ by specifying the probabilistic metric $ \mathcal{g}_x $, which we can very well do. Of course, the $ P^{g} $ will be very rough, because it will be only able to measure the sets of all metric field realizations with a common metric tensor in point $ x $. This is fine as long as we do not care about the metric in any other point.\\

We can ask how well it is possible to describe a p.st. $ \mathcal{M} $ by multiple probability measures, e.g. those given by $ \mathcal{g}_x $ for all points of the manifold. The brief answer is brought hereafter.\\

Given a p.st. $ \mathcal{M} $, one can easily reduce it to a l.p.st. $ \mathcal{M}_{\text{loc.}} $ using the discussed fragmentation $ y(g) = g(x) $ (although only if we have a sigma field and measure on $ \mathscr{R} $ defined within $ \mathcal{g} $ that can be used for the task, that is, they must contain the information about the probability of each metric tensor realization in each point) separately for all $ x $ of $ M $, and defining the corresponding probabilistic tensors $ \mathcal{g}_x $.\\

Given a l.p.st. $ \mathcal{M}_{\text{loc.}} $, it determines a class of sigma fields $ \mathscr{P}^{g,x} $ and measures $ P^{g,x} $ on $ \mathscr{R} $. All of them are very imperfect, but one can try to combine the information to get an estimate of the probability of an \emph{arbitrary}\footnote{An arbitrary subset is the one we do not have any reasonable measure for.} subset $ X $ of $ \mathscr{R} $.\\

Here is a simple illustration. Let us consider only points $ x, y, z $ of $ M $ (for the sake of the example, we will just ignore all other points) and two options for the metric tensor $ g_{x1}, g_{x2} $ in $ T^{g(x)} $, $ g_{y1}, g_{y2} $ in $ T^{g(y)} $ and $ g_{z1}, g_{z2} $ in $ T^{g(z)} $. Suppose we have the measure $ P^{g,x}(g_x) $ describing the probability of a realization $ g_x $ in $ T^{g(x)} $ and in the other points alike. Further, suppose we have the detailed, global measure $ P^{g} $ (which is no problem in this discrete case). Instead of $ P^{g,x}(g_{x1}) $, we will write just $ P^{x}(g_1) $; and so on. Take e.g.
\begin{equation}
\begin{gathered}
P^{g} \left( \left \lbrace \begin{pmatrix}
g_{x1} \\
g_{y1} \\
g_{z1}
\end{pmatrix}, \begin{pmatrix}
g_{x2} \\
g_{y2} \\
g_{z2}
\end{pmatrix} \right\rbrace \right) = \\ = \text{probability of realization } g_1 \text{ or } g_2 = \\ = P^{x}(g_1) P^{y}(g_1) P^{z}(g_1) + P^{x}(g_2) P^{y}(g_2) P^{z}(g_2).
\end{gathered}
\end{equation}
We can approach this with the partial measures by computing
\begin{equation}
\begin{gathered}
P^{x} \left( \lbrace g_1 , g_2 \rbrace \right) P^{y} \left( \lbrace g_1 , g_2 \rbrace \right) P^{z} \left( \lbrace g_1 , g_2 \rbrace \right) = \\ = \left( P^{x}(g_1) + P^{x}(g_2) \right) \left( P^{y}(g_1) + P^{y}(g_2) \right) \left( P^{z}(g_1) + P^{z}(g_2) \right) = \\ = P^{x}(g_1) P^{y}(g_1) P^{z}(g_1) + P^{x}(g_2) P^{y}(g_2) P^{z}(g_2) + \text{other terms}.
\end{gathered}
\end{equation}
The conclusion is
\begin{equation}
\text{product of partial measures} \geq \text{global measure}.
\end{equation}

We remark that this may work the same for the general case of $ \mathcal{M}_{\text{loc.}} $. Let $ X $ be a subset of $ \mathscr{R} $ and let $ A_x $ be the corresponding subset of $ T^{g(x)} $. Then we define
\begin{equation}
\xi(X) = \exp \int_M \ln \left( \int_{A_x} dH \right) dV,
\end{equation}
where $ dV = \sqrt{\vert g \vert} ~ d^{4} x $. Further, we define $ \xi(X) = 0 $ whenever $ \int_{A_x}  dH = 0 $ for some $ x $. This is an estimate of the probability of a set of metric field realizations $ X $, which could give an upper bound on what we would expect to be the unknown, corresponding detailed measure on $ \mathscr{R} $.\\

Despite its relative simplicity in comparison with p.st., locally probabilistic spacetime can be useful. It is particularly convenient to work with l.p.st. when we are only interested in local information, i.e. the value of the metric tensor in some point of the manifold. We could go further and define analogues of l.p.st. for \emph{derivatives} of metric (provided they exist), which would enable us to give local probabilistic description to objects like Christoffel symbols or the Riemann and Ricci tensors, without the use of a rather more complicated p.st. The downside is that the derivatives of metric are not tensors, so we would have to do it non-covariantly. However, it is clear that the more derivatives we would locally define, the more precise (in some sense) would our description be. Once we have defined locally all the derivatives in a single point (provided all the corresponding distribution functions) of an analytic manifold, it would uniquely fix the measure on the analytic subset of $ \mathscr{R} $ via Taylor series.

\subsubsection*{Composite Probabilistic Spacetime}
Maybe the easiest way to define a p.st. is as a probabilistic \emph{composite} of classical spacetimes $ \mathcal{C} = \lbrace (M, g_j, \alpha_j) \rbrace_{j \in J } $. It is a set of spacetimes that share the manifold $ M $ but differ in metric tensor fields $ g_j $, together with the real numbers $ \alpha_j $. The character of the composite is determined by the character of the index set $ J $. If $ J $ is countable, the numbers $ \alpha_j \in [0,1] $ are the probabilities of the respective spacetimes $ (M, g_j) $. It is required that $ \sum_j \alpha_j = 1 $ and the composite is called \emph{discrete}. If the set $ J $ is homomorphic to $ \mathbb{R}^{N} $, the sum becomes an integral and we have to supply the $ N $-dimensional probability density function $ \alpha_j \equiv \alpha(j) $ whose range is $ \mathbb{R}_0^{+} $ and $ j $ stands for the $ N $-tuple $ (j_1, ..., j_N) $. It is required that $ \int \alpha(j) dj = 1 $ and the composite is called \emph{$ N $-parametric}.\\

Let $ C $ be a subset of $ \mathscr{R} $ containing the metric tensor fields of a composite $ \mathcal{C} $, and let $ \mathscr{P}^{g}_{C} $ be the perfect sigma field on it. Then $ \mathcal{C} $ together with its $ P^{g}_{C} $ determine a measure on $ \mathscr{R} $, namely $ P^{g}(X) = P^{g}_{C}(X \cap C) $\footnote{For example, if $ C $ is composed of two realizations $ g_1, g_2 $, then $ P^{g}(X) $ is equal to 1 if $ X $ contains both $ g_1 $ and $ g_2 $, 0 if it contains neither, and it is equal to the probability $ P^{g}_{C}(g_1) $ or $ P^{g}_{C}(g_2) $ if it contains only $ g_1 $ or only $ g_2 $, respectively.}. This means that especially for $ g_Z \in Z $, $ Z = \mathscr{R} \smallsetminus C $, it is $ P^{g}(\lbrace g_Z \rbrace) = 0 $ and the same for all subsets of $ Z $. In other words, we made the measure trivial on $ Z $, because it simply ignores its elements. One can easily see that if $ P^{g}_C $ is a measure on $ C $, then $ P^{g} $ is a measure on $ \mathscr{R} $. We remark that this construction is more advantageous than a similar construction using fragmentation, which would map all $ g_Z \in Z $ to one zero-measure element $ z \in S $. This would not take full advantage of the zero measure of subsets of $ Z $, since it would not enable us to measure sets containing both elements of $ C $ and (some but not all) elements of $ Z $.\\

The measure $ P^{g}_C $ can be defined in accordance with the numbers $ \alpha_j $. In the discrete case, we write
\begin{equation}
P^{g}_C(U) = \sum_{j \in J} \alpha_j \chi_{U}(g_j) \label{P_composite_disc}
\end{equation}
for all $ U $ in $ \mathscr{P}^{g}_{C} $ and $ g_j $ in $ C $. For a set $ S $, we use the characteristic function
\begin{equation} \label{char_function}
\chi_{S}(x) = \begin{cases}
    1  & \quad \text{if } ~ x \in S\\
    0  & \quad \text{otherwise}.\\
  \end{cases}
\end{equation}
Analogically for the $ N $-parametric case:
\begin{equation}
P^{g}_C(U) = \int_J \alpha(j) \chi_{U}(g_j) dj \label{P_composite_cont}.
\end{equation}

Let us check that $ P^{g}_C $ defined by (\ref{P_composite_disc}) or (\ref{P_composite_cont}) are indeed probability measures on $ C $. In the discrete case, we have the following:
\begin{enumerate}[label=\roman*., itemsep=-0.2mm]
\item For the whole $ C $ and for empty set, it holds
\begin{equation}
\begin{aligned}
& P^{g}_C(C) = \sum_j \alpha_j \chi_{C}(g_j) = \sum_j \alpha_j = 1, \\ & P^{g}_C(\varnothing) = \sum_j \alpha_j \chi_{\varnothing}(g_j) = 0.
\end{aligned}
\end{equation}
\item If $ \lbrace A_n \rbrace $ is a sequence of pairwise disjoint sets in $ \mathscr{P}^{g}_C $, then
\begin{equation}
P^{g}_C \left( \bigcup_{n = 1}^{\infty} A_n \right) = \sum_j \alpha_j \chi_{ \bigcup_{n = 1}^{\infty} A_n }(g_j) = \sum_{n = 1}^{\infty} \sum_j \alpha_j \chi_{A_n}(g_j) = \sum_{n = 1}^{\infty} P^{g}_C(A_n).
\end{equation}
\end{enumerate}
Thus, according to (\ref{def:probability_space}), $ P^{g}_C $ is a probability measure on $ C $. The same can be shown for the continuous case, which is completely analogical. The fact that one can assign probabilities to a set of spacetimes (their manifolds must be diffeomorphic!) and create a p.st. is trivial yet important, because it is a very useful technique to start with. In fact, composite probabilistic spacetime is the first, and probably the last example of p.st. for which we have a measure associated with the perfect sigma field.\\

We amend the natural definition:

\begin{definition}
We say that a p.st. $ \mathcal{M} $ is \emph{composite}, if there exists a composite $ \mathcal{C} $, such that $ \mathcal{M} $ is obtained from $ \mathcal{C} $ via (\ref{P_composite_disc}) (discrete case) or (\ref{P_composite_cont}) ($ N $-parametric case).
\end{definition}

Let us add a remark on coordinates. Given a p.st. $ \mathcal{M} = (M, \mathcal{g}) $, we may work in any coordinates in which the realizations $ g $ of $ \mathcal{g} $ can be expressed. If the p.st. is provided in the form of a composite $ \mathcal{C} = \lbrace (M, g_j, \alpha_j) \rbrace_{j \in J } $, we typically have a separate set of coordinates for every component $ (M, g_j) $. However, providing separately the metrics of these components (with their respective probabilities $ \alpha_j $) is insufficient! The composite is by definition built on a single manifold, and when speaking about the metric tensor, we should be able to determine at which point of the manifold it is defined. That is why in the case described we have to provide a unique way to assign a manifold point to the coordinate values, and if the coordinates come from two distinct spacetimes, we have to identify their manifolds. In other words, for every two components equipped with their own coordinates, we have to provide a diffeomorphism $ \phi $ mapping points of one component to points of the other. Then, we can pass between different tensor representations via the standard pushforward and pullback maps. This enables expressing all $ g $'s in the same coordinates.

\subsection{Relationship with PM Spaces}
We would like to know whether the probabilistic metric tensor $ \mathcal{g}_p $ in a point $ p $ of an l.p.st. $ \mathcal{M}_{\text{loc.}} = (M, \mathcal{g}_x) $ can be related to the metric of a PM space $ (S, \mathscr{F}, \tau) $ from Definition \ref{def:PMspace}. For that purpose, let us consider the function $ F(u,v) $ from $ T_{p} M \times T_{p} M $ into $ \Delta $ such that $ F(u,v) = F_{uv}(x) $,
\begin{equation}
F_{uv}(x) = \int \chi_{ (-\infty, x ) } \left( d(u,v) \right) ~ \mathrm{d} H, \label{Fuv}
\end{equation}
where we use the ``distance'' induced by metric,
\begin{equation}
d(u,v) \equiv \Vert u - v \Vert \equiv \sqrt{\vert g_{p}(u-v,u-v) \vert}.
\end{equation}
The integral is defined by the measure of $ \mathcal{g}_p $ (typically specified by the $ m $-distribution function $ H $), as has been discussed before, and sums over the realizations of the metric tensor $ g_p $ in point $ p $. We took advantage of (\ref{char_function}) once again. We note that since $ d(u,v) $ is a scalar, the function $ F_{uv} $ is covariant.\\

Recalling Definition \ref{def:df}, we see that $ F_{uv} $ is indeed a distribution function. It has a straightforward interpretation
\begin{equation}
F_{uv}(x) = \text{probability that} ~ d(u,v) < x. \label{F_interpretation}
\end{equation}

In the next, we assume that the probabilistic metric tensor $ \mathcal{g}_p $ is \emph{causally straight}, meaning that if $ g_{p}(a, b) \lesseqgtr 0 $ for two vectors $ a, b $ and some realization $ g_{p}$ of $ \mathcal{g}_p $, then it holds for all realizations. Under this assumption, we can take a look at a triple of vectors $ u, w, v $, such that $ g_p(v-u,v-u) \geq 0 $ and $ g_p(v-u,v-u) = 0 $ iff $ u = v $, and so forth for all couples. In this case, the triangle inequality in $ T_{p} M $ has the usual form \cite{Naber2012}
\begin{equation}
\Vert (w - u) + (v - w) \Vert \leq \Vert w - u \Vert + \Vert v - w \Vert
\end{equation}
or in other words
\begin{equation}
d(u,v) \leq d(u,w) + d(w,v) \label{triangle}
\end{equation}
and, by our assumption, it holds for all realizations $ g_p $ of $ \mathcal{g}_p $.\\

Now, the probability density function $ f_{uv} $ of the sum of two random variables described by their probability density functions $ f_{uw} $ and $ f_{wv} $ is their convolution,
\begin{equation}
f_{uv}(x) = (f_{uw} \ast f_{wv})(x) \equiv \int_{\mathbb{R}} f_{uw}(y) f_{wv}(x-y) ~ \mathrm{d}y.
\end{equation}
So, if we want to compute
\begin{equation}
F_{uw + wv}(x) = \text{probability that} ~ d(u,w) + d(w,v) < x, \label{F_+interpretation}
\end{equation}
we simply write
\begin{equation}
F_{uw + wv}(t) = \int_{-\infty}^{t} \int_{\mathbb{R}} f_{uw}(y) f_{wv}(t'-y) ~ \mathrm{d}y \mathrm{d}t' = \int_{\mathbb{R}} f_{uw}(y) F_{wv}(t-y) ~ \mathrm{d}y. \label{Fuwwv}
\end{equation}

All we have to do is compare (\ref{F_interpretation}) with (\ref{F_+interpretation}), keeping in mind (\ref{triangle}). We find out that
\begin{equation}
F_{uv}(x) \geq F_{uw + wv}(x). \label{F_triangle}
\end{equation}
This is the triangle inequality that we need in order to satisfy the definition of a triangle function. We thus define our triangle function candidate according to (\ref{Fuwwv}) as
\begin{equation}
\tau(F,G)(x) = \int_{\mathbb{R}} \frac{\mathrm{d} F(y)}{\mathrm{d}y} G(x - y) ~ \mathrm{d}y \label{eq:tau}
\end{equation}
for $ F, G \in \Delta^{+} $. We shall examine a few of $ \tau $'s properties:
\begin{enumerate}[label=\roman*., itemsep=-0.2mm]
\item Commutativity.
\begin{equation}
\begin{aligned}
\tau(F,G)(x) &= \bigg[ F(y) G(x - y) \bigg]_{-\infty}^{\infty} - \int_{\mathbb{R}} F(y) \frac{\mathrm{d} G(x - y)}{\mathrm{d}y} ~ \mathrm{d}y = \\ &= - \int_{\mathbb{R}} F(x - z) \frac{\mathrm{d} G(z)}{\mathrm{d}z}(-) ~ \mathrm{d}z = \tau(G,F)(x).
\end{aligned}
\end{equation}

\item Associativity.
\begin{equation}
\begin{aligned}
\tau(\tau(F,G),H)(x) &= \int_{\mathbb{R}} \frac{\mathrm{d}}{\mathrm{d}y} \left[ \int_{\mathbb{R}} \frac{\mathrm{d} F(z)}{\mathrm{d}z} G(y - z) ~ \mathrm{d}z \right] H(x - y) ~ \mathrm{d}y
= \\ &= \int_{\mathbb{R}} \int_{\mathbb{R}} \frac{\mathrm{d} F(z)}{\mathrm{d}z} \frac{\mathrm{d} G(t)}{\mathrm{d}t} H (x - t - z) ~ \mathrm{d}z \mathrm{d}t
= \\ &= \int_{\mathbb{R}} \frac{\mathrm{d}F(z)}{\mathrm{d}z} \int_{\mathbb{R}} \frac{\mathrm{d} G(t)}{\mathrm{d}t} H(x - t - z) ~ \mathrm{d}z \mathrm{d}t
= \\ &= \tau(F, \tau(G,H))(x).
\end{aligned}
\end{equation}

\item Non-decreasing property. Assume $ G > H $, i.e., $ G(x) > H(x) $ for all $ x $. Then
\begin{equation}
\begin{aligned}
\tau(F,G)(x) &= \int_{\mathbb{R}} \frac{\mathrm{d} F(y)}{\mathrm{d}y} G(x - y) ~ \mathrm{d}y > \\ & > \int_{\mathbb{R}} \frac{\mathrm{d} F(y)}{\mathrm{d}y} H(x - y) ~ \mathrm{d}y = \tau(F,H)(x).
\end{aligned}
\end{equation}

\item Identity embodied by $ \varepsilon_0 $ (Def. \ref{def:epsilon}).
\begin{equation}
\begin{aligned}
\tau(F,\varepsilon_0)(x) &= \int_{\mathbb{R}} \frac{\mathrm{d} F(y)}{\mathrm{d}y} \varepsilon_0(x - y) ~ \mathrm{d}y = - \int_{\mathbb{R}} F(y) \frac{\mathrm{d} \varepsilon_0(x - y)}{\mathrm{d}y} ~ \mathrm{d}y = \\ &= \int_{\mathbb{R}} F(y) \frac{\mathrm{d} \varepsilon_0(x - y)}{\mathrm{d}x} ~ \mathrm{d}y = \frac{\mathrm{d}}{\mathrm{d}x} \int_{\mathbb{R}} F(y) \varepsilon_0(x - y) ~ \mathrm{d}y = \\ &= \frac{\mathrm{d}}{\mathrm{d}x} \int_{-\infty}^{x} F(z) ~ \mathrm{d}z = F(x).
\end{aligned}
\end{equation}
\end{enumerate}

From these four points, it follows that $ \tau $ is a triangle function (Def. \ref{def:triangle_function}).\\

To recap, we have the set $ T_{p} M_{\text{pos.}} $ (a vector tangent space restricted so that $ g_p(v-u,v-u) \geq 0 $ and $ g_p(v-u,v-u) = 0 $ iff $ u = v $ for all $ u,v $ in $ T_{p} M_{\text{pos.}} $) and a function $ F(u,v) = F_{uv}(x) $ given by (\ref{Fuv}) (for a causally straight metric) from $  T_{p} M_{\text{pos.}} \times T_{p} M_{\text{pos.}} $ into $ \Delta $. Under our assumption about the vectors (positive ``distances''), $ \Delta $ can be narrowed down to $ \Delta^{+} $. The following are true:
\begin{enumerate}[label=\roman*., itemsep=-0.2mm]
\item Identical points.
\begin{equation}
F(u,u) = \int \chi_{ (-\infty, x ) } \left( 0 \right) ~ \mathrm{d} H = \varepsilon_0.
\end{equation}
\item Suppose $ u \neq v $. Then under our assumptions (and only under them!), $ d(u,v) > 0 $, and clearly
\begin{equation}
F(u,v) = \int \chi_{ (-\infty, x ) } \left( d(u,v) \right) ~ \mathrm{d} H \neq \varepsilon_0.
\end{equation}
\item By definition (and thanks to the symmetry of $ d(u,v) $), $ F(u,v) = F(v,u) $.
\item We have found a triangle function $ \tau $ such that, according to (\ref{F_triangle}),
\begin{equation}
F(u,v) \geq \tau( F(u,w), F(w,v) ).
\end{equation}
\end{enumerate}
Hence, in the sense of Definition \ref{def:PMspace}, the triple $ (T_{p} M_{\text{pos.}}, F, \tau) $ is a PM space.\\

Comparing our triangle function with that of Def. (\ref{def:triangle_function_classification}), we see that were it not for the derivative in the integrand in (\ref{eq:tau}), $ (T_{p} M_{\text{pos.}}, F, \tau) $ would be a Wald space.\footnote{In fact, Wald also describes his convolution to be the distribution function of the sum \cite{Wald1943}; this---as we believe---misinterpretation is later repeated by Schweizer and Sklar \cite{Schweizer1960}. Nevertheless, our form of the triangle function is of course just as viable as the Wald's.}\\

Keeping a causally straight metric, we could progress analogically in the case of \emph{timelike} separation. Namely, for three vectors $ u,v,w $ such that $ v-u, w-u, v-w $ are all timelike vectors of the same orientation (i.e., $ g(w-u, v-w) < 0 $), one has the Reversed Triangle Inequality\footnote{To be found in \cite{Naber2012}, p. 44.}
\begin{equation}
d(u,v) \geq d(u,w) + d(w,v)
\end{equation}
(and the equality holds iff $ w-u $ and $ v-w $ are linearly dependent). For our purposes, we multiply the inequality by $ -1 $,
\begin{equation}
- d(u,v) \leq - d(u,w) - d(w,v)
\end{equation}
which gives it the form encountered in the previous case. One would then simply take
\begin{equation}
F_{uv}(x) = \int \chi_{ (-\infty, x ) } \left( - d(u,v) \right) ~ \mathrm{d} H,
\end{equation}
and proceed the same way.\\

For other combinations of vectors, as well as for a general probabilistic metric tensor, the analysis would presumably get more complicated, but also less interesting, because the PM space is clearly not designed for the indefinite metric. At this point, we see an asset in showing the close affinity of the concept to our structures and finding their exact relation, as has been done.\\

\subsection{Test Point Particle}
Let us assume we have a well defined p.st., e.g. in the form of a composite. To gain some basic understanding of its structure and behaviour, it would be of benefit to study the kinematical properties of the p.st., that is, the movements of test particles. However, it is not a priori clear how the kinematics should actually be treated. In the following, we will look at some possibly useful generalizations of the test particle, which could be used for the analysis.\\

First of all, we want to employ a test point particle. The movement of the particle is specified by a worldline in the p.st. The complication is that this worldline, up to special cases, is not a geodesic in more than one realization of the p.st. In this sense, there are no classical geodesics. We thus need to consider either multiple worldlines associated with different realizations or an alternative definition of the geodesic. An implementation of the latter and the simplest generalization of kinematics for the p.st. is what one would aptly call a \emph{mean geodesic}. It shall be defined as a worldline $ \gamma(\lambda) $ extremizing the mean action $ \bar{S} $. Recalling the action for a particle of mass $ m $ along a segment of an arbitrary worldline parametrized by $ \lambda $,
\begin{equation}
S_g = - m \int_{\lambda_0}^{\lambda_1} \sqrt{- g_{ \mu \nu}(x) \frac{dx^{\mu}}{d\lambda} \frac{dx^{\nu}}{d\lambda} } ~ d \lambda,
\end{equation}
we can compute the mean action
\begin{equation}
\bar{S} = \int S_g ~ dG.
\end{equation}
It is then only natural to define the mean geodesic by means of the variation principle,
\begin{equation}
\delta \bar{S} = 0. \label{variation_mean_worldline}
\end{equation}

One can exploit (\ref{variation_mean_worldline}) to get the equations of motion for the mean geodesic. By a simple computation, one reproduces the usual equation of geodesic, only weighted by the probabilities of distinct realizations,
\begin{equation}
\int \left[ \frac{d^2 x^{\beta}}{d \tau^2} + \Gamma^{\beta}_{ ~ \mu \nu} \frac{d x^{\mu}}{d \tau} \frac{d x^{\nu}}{d \tau} \right] dG = 0. \label{mean_geodesic}
\end{equation}

We must not forget that the action associated with a worldline is a random variable. Its distribution function can be written down in usual manner,
\begin{equation}
F_S (x) = \int \chi_{ (-\infty, x ) } \left( S_g \right) ~ dG.
\end{equation}
Apart from the mean action $ \bar{S} $, we can compute higher moments like the variance
\begin{equation}
\sigma_S^{2} = \int (x - \bar{S})^{2} f_S(x) ~ dx,
\end{equation}
where $ f_S = \int \delta \left( S_g - x \right) ~ dG $ is the probability density function associated with $ F_S (x) $. In other words,
\begin{equation}
\sigma_S^{2} = \int (x - \bar{S})^{2} \int \delta \left( S_g - x \right) ~ dG dx = \int (S_g - \bar{S})^{2} ~ dG.
\end{equation}
These very simple characteristics can give us a hint on the spatial spread of probable worldlines. Let us define \emph{worldline within the range of variance} to be a worldline $ \gamma $ connecting two points $ A, B $ of the manifold, satisfying
\begin{equation}
\left( \bar{S}(\gamma) - \bar{S}(\text{mean geodesic}) \right)^{2} \leq \sigma_S(\text{mean geodesic})^{2}.
\end{equation}
Here, $ \bar{S}(\gamma) $ is the mean action of the worldline $ \gamma $ between $ A $ and $ B $, which is compared to the mean action of the mean geodesic (which is extremized)  between the same points. If the difference squared is smaller than or equal to the variance of the mean geodesic itself, it means that the worldline $ \gamma $ is in some sense kinematically close to being geodesic.\\

If a worldline is within the range of variance, we argue that its probability of being a geodesic should be comparable to that of the mean geodesic. Based on such criteria, we can determine the probable range of physical worldlines for the classical test particle.

\subsection{Probabilistic Particle}
The goal of this section is to introduce a generalized description of the test point particle, such that it would better fit the statistical nature of the model. This is important for the interpretation of a p.st., understanding the information carried by the probabilistic metric and the study of a p.st.'s overall kinematical properties. The probabilistic particle should be allowed to exist in many different states, weighted by a distribution of probability. In classical case, the state is given by position and velocity, i.e. a point in the tangent bundle. That is where we define the probabilistic particle.\\

For start, we will work in a spacetime $ M $. Let $ A(x,u) $ be a scalar function on the tangent bundle $ TM $. This function shall represent the probability density of a particle being found at a specified position with a specified velocity. The couple $ (x,u) $ stands for a point of $ TM $, and is given by a point $ x $ of $ M $ and the \emph{proper} velocity $ u $ of $ T_x M $. Whenever we will speak about the position $ x $ or the velocity\footnote{It would be more correct to write $ u_x $. We will use this notation where we want to stress out the affiliation of the velocity to its tangent space.} $ u $, we will mean the \emph{projections} from $ TM $ to $ M $ or $ T_x M $, respectively.\\

We have to provide a proper normalization for $ A $. At a given instance of an arbitrary coordinate time, we would like the probability of finding the particle somewhere on the spacelike hypersurface to be equal to one. We thus need to integrate over the spacelike hypersurface, and over all velocities. The former can be done in a conventional way, writing
\begin{equation}
\int_{\Sigma} A(x,u) ~ u^{\mu} d S_{\mu},
\end{equation}
where $ \Sigma $ is the hypersurface in question and $ d S_{\mu} $ is its normal volume element. The contraction with velocity establishes some kind of ``measure of presence'' of the particle, since it locally projects the velocity to $ \Sigma $. One can also think of the vector $ A(x,u) u^{\mu} $ as the probability flux, and its projection to the hypersurface establishes the probability of the particle passing through the volume element.\\

To show why this integral suits our purpose, let us give a trivial example. Consider a static probabilistic particle in the Minkowski spacetime, uniformly distributed in a unit cube $ x \in [x_A, x_B] $, where $ x_B - x_A = 1 $ and the same for $ y $ and $ z $. If we integrate the function\footnote{We suppress the velocity argument so we do not have to deal with a Dirac delta; effectively, our $ A(x) $ is already integrated over velocities. Also, to be perfectly correct, one should differentiate between the manifestly covariant function $ A(x,u) $ on the tangent bundle (or just on the manifold, as in this case) and its coordinate representation $ A(x^{\mu}, u^{\mu}) $. Here, the distinction is done by writing arguments explicitly.} $ A(x,y,z) $ over the hypersurface $ \Sigma $ given by $ t = 0 $, we get
\begin{equation}
\int_{\Sigma} A(x,y,z) ~ u^{\mu} d S_{\mu} = \int A(x,y,z) ~ (-) ~ dx ~ dy ~ dz = -1,
\end{equation}
because $ u^{\mu} = (1,0,0,0) $ and the normal to $ \Sigma $ is also $ n^{\mu} = (1,0,0,0) $, so their scalar product in the flat metric is just $ -1 $. Within the calculation, we anticipated the correct normalization $ A = 1 $ inside the cube.

Now let us choose another hypersurface, e.g. $ \tilde{\Sigma} $ given by $ t' = 0 $ in a frame boosted in $ t-x $ plane. The new frame shall be
\begin{equation}
t' = \gamma (t - vx), \qquad x' = \gamma (x - vt), \qquad y' = y, \qquad z' = z,
\end{equation}
and the velocity of the previously-static particle is $ u'^{\mu} = (\gamma, - \gamma v) $. The normal to $ \tilde{\Sigma} $ is (in the new frame) again $ \tilde{n}'^{\mu} = (1,0,0,0) $. Because the tilted hypersurface $ \tilde{\Sigma} $ is given by $ t' = \gamma (t - vx) = 0 $, it holds $ t_{cube} = v x_{cube} $ and the positions of the cube's vertices are subjected to dilatation as
\begin{equation}
x'_A = \gamma (x_A - vt) = \gamma (x_A - v^{2} x_A) = \gamma x_A (1 - v^{2}) = x_A \gamma^{-1},
\end{equation}
and the same for $ x_B $. Now, the integral is
\begin{equation}
\begin{aligned}
\int_{\tilde{\Sigma}} A(x,y,z) ~ u^{\mu} d S_{\mu} = \int A(x',y',z') ~ (- \gamma) ~ dx' ~ dy' ~ dz' = \\ = \int_{ x_A \gamma^{-1}}^{x_B \gamma^{-1}} (- \gamma) ~ dx' = - \gamma \gamma^{-1} (x_B - x_A) = - 1.
\end{aligned}
\end{equation}
The result for the tilted hypersurface is the same, which is exactly what we need.\\

To find a convenient form of the integral over velocities, let us look closer at the tangent space in a point $ p $ of $ M $. The proper velocities $ u $ lay on the \emph{upper hyperboloid} (briefly, \emph{u.h.}), which is a subset of the tangent space $ T_p M $ such that every $ u $ of the u.h. is future-directed and satisfies $ u^{\mu} u_{\mu} = -1 $. For now, let us assume that the metric in $ T_p M $ is just $ \mathrm{diag}(-1,1,1,1) $. We will define a covariant integral over u.h. as an integral over the analogical upper hyperboloid\footnote{In what follows, we will not differentiate between the two.} in Minkowski space associated\footnote{There is a canonical identification of Minkowski space with its own vector tangent spaces in every point \cite{Lee2012} and we can further identify the prototypic tangent space of Minkowski with our $ T_p M $---they have the very same structure.} with $ T_p M $. An example of such an integral is
\begin{equation}
\int_{\mathrm{u.h.}} u^{\mu} dB_{\mu}, \label{int_uh}
\end{equation}
where $ u $ is the vector associated with a point in Minkowski space and $ dB_{\mu} $ is the normal volume element of the upper hyperboloid in Minkowski space.\\

At first sight, it may not be clear whether the integral (\ref{int_uh}) diverges or not. We shall briefly inspect its behaviour. In the Cartesian coordinates of Minkowski, the hyperboloid is given by $ -t^{2} + x_1^{2} + x_2^{2} + x_3^{2} = -1 $ (we write lower indices for aesthetic reasons). Let us denote $ x_1^{2} + x_2^{2} + x_3^{2} \equiv r^{2} $, so $ t^{2} = 1 + r^{2} $ and for the u.h. we have $ t = \sqrt{1 + r^{2}} $. It may be parametrized by $ x_1^{\mathrm{u.h.}}, x_2^{\mathrm{u.h.}}, x_3^{\mathrm{u.h.}} $ equal to $ x_1, x_2, x_3 $, respectively. One easily finds out that the unit normal $ n $ to the u.h. has components
\begin{equation}
n^{0} = \frac{t^{2}}{t^{2} - r^{2}}, \qquad n^{i} = \frac{x^{i} t}{t^{2} - r^{2}}.
\end{equation}
For a velocity vector associated with a point in Minkowski space $ u $, which lies on the u.h., it follows that
\begin{equation}
\eta_{\mu \nu} u^{\mu} n^{\nu} = -t.
\end{equation}
The three-dimensional induced metric on the u.h.
\begin{equation}
h_{a b} = \eta_{\alpha \beta} \frac{\partial x^{\alpha}}{\partial x^{a}_{\mathrm{u.h.}}} \frac{\partial x^{\beta}}{\partial x^{b}_{\mathrm{u.h.}}}
\end{equation}
has components
\begin{equation}
h_{a b} = \frac{1}{t^{2}} \begin{pmatrix}
  t^{2} - x_1^{2} & - x_1 x_2  & - x_1 x_3  \\
   & t^{2} - x_2^{2} & - x_2 x_3 \\
  \mathrm{sym.} &  & t^{2} - x_3^{2}
\end{pmatrix}.
\end{equation}
Its determinant is simply
\begin{equation}
\vert h \vert = \frac{1}{t^{2}}.
\end{equation}
Now, we can write
\begin{equation}
\int_{\mathrm{u.h.}} u^{\mu} dB_{\mu} = \int u^{\mu} n_{\mu} \sqrt{\vert h \vert} ~ d^{3} x_{\mathrm{u.h.}} = \int -1 ~ d^{3} x_{\mathrm{u.h.}}. \label{velocities_volume_integral}
\end{equation}
In conclusion, we have a covariant integral over the space of proper velocities which behaves like a simple three-dimensional volume integral. Now we go back to our correspondence of $ T_p M $ with the Minkowski space, in which we can define a scalar function $ f(x) $ on the u.h. This gives a clear sense to the integral
\begin{equation}
\int_{\mathrm{u.h.}} f(u) ~ u^{\mu} dB_{\mu}.
\end{equation}
To compute it, we represent the abstract function $ f(u) $ by a function $ f(u^{\mu}) $ of the coordinate values. Then we proceed as before,
\begin{equation}
\begin{aligned}
\int_{\mathrm{u.h.}} f(u) u^{\mu} dB_{\mu} & = \int f(u^{\mu}) u^{\mu} n_{\mu} \sqrt{\vert h \vert} ~ \delta(u^{t} - u^{t}(...)) ~ d^{4} u = \\ & = \int - f(u^{\mu}) ~ \delta(u^{t} - u^{t}(...)) ~ d^{4} u.
\end{aligned}
\end{equation}
This is a prescription for Cartesian coordinates. In contrast with (\ref{velocities_volume_integral}), there is a function which generally takes all proper velocity components, including the coordinate time---we have to employ the u.h. constraint $ \delta(u^{t} - u^{t}(...)) $, where $ u^{t}(...) $ stands for $ u^{t} $ as a function of the coordinate values on the u.h.\\

We need to generalize the integral to the case of an arbitrary metric $ g_{\mu \nu}(p) $ in a point $ p $. According to the Local Flatness Theorem\footnote{It can be found in \cite{Schutz2009}, p. 149}, it is possible to choose a new basis in $ T_p M $ such that
\begin{equation}
g_{\alpha \beta}(p) = \Lambda^{\mu}_{~ \alpha}(p) \Lambda^{\nu}_{~ \beta}(p) \eta_{\mu \nu}.
\end{equation}
We note that if this was a change of coordinates---which is not really the case---the matrix would be
\begin{equation}
\Lambda^{\mu}_{~ \alpha}(p) = \frac{\partial \tilde{x}^{\mu}}{\partial x^{\alpha}}.
\end{equation}
If $ u_p^{\mu} $ are the components of the proper velocity $ u_p $ (of $ T_p M $) in arbitrary coordinate system on $ M $, then its components in the new local basis are
\begin{equation}
\tilde{u}_{p}^{\mu} = \Lambda^{\mu}_{~ \alpha}(p) u_{p}^{\alpha}.
\end{equation}
This gives us a rule for integrating a function $ f(x,u) $ on the tangent bundle over proper velocities. We have to adapt the basis of the tangent space in every point, and in return we may use the result for the integral in Minkowski spacetime:
\begin{equation}
\begin{aligned}
& \int_{\mathrm{u.h.}} f(x,u) ~ u^{\mu} dB_{\mu} = \\
& = \int - f(x^{\mu},u_{x}^{\mu}) ~ \delta(\tilde{u}_{x}^{t} - \tilde{u}_{x}^{t}(...)) ~ d^{4}\tilde{u}_x \equiv \\
& \equiv \int - \tilde{f}(x^{\mu},\tilde{u}_{x}^{\mu}) ~ \delta(\tilde{u}_{x}^{t} - \tilde{u}_{x}^{t}(...)) ~ d^{4}\tilde{u}_x = \\
& = \int - \tilde{f}(x^{\mu},\Lambda^{\mu}_{~ \alpha}(x) u_x^{\alpha}) ~ \delta( \Lambda^{t}_{~ \alpha}(x) u_x^{\alpha} - \Lambda^{t}_{~ \alpha}(x) u_x^{\alpha}(...)) ~ \left\vert \Lambda^{\nu}_{~ \beta}(x) \right\vert ~ d^{4}u_x. \label{vel_int}
\end{aligned}
\end{equation}
Here, $ f(x^{\mu},u_{x}^{\mu}) $ is the coordinate representation of $ f(x,u) $ taking the coordinate values. One takes $ u_x^{\alpha} $ for $ u_x^{\alpha}(...) $ if $ \alpha = 1,2,3 $ and $ u_x^{t} $ as a function of the others for $ u_x^{t}(...) $. And again, this form can be used only if the velocities are given in Cartesian coordinates.\\

We eventually have at our disposal the integral over proper velocities. For brevity, we will occasionally write $ dB $ instead of $ u^{\mu} dB_{\mu} $. Now, it is time to address the definition of the probabilistic particle.

\begin{definition} \label{def:ppd}
\emph{Probabilistic particle density} (briefly, \emph{p.p.d.}) is a scalar function $ A(x, u) $ on the tangent bundle $ TM $ of a spacetime manifold $ M $ with values from $ \mathbb{R}_0^{+} $, such that for every partial Cauchy surface\footnote{A partial Cauchy surface is an acausal set without an edge. For more information, we refer to \cite{Joshi1993}, p. 119.} $ \Sigma $ it satisfies the normalization condition
\begin{equation}
\int_{\Sigma} \int_{\mathrm{u.h.}} A(x, u) ~ u^{\mu} dB_{\mu} ~ u^{\nu} d S_{\nu} = 1. \label{A_completeness_relation}
\end{equation}
\end{definition}

\subsubsection*{Local Normalization Condition}
We will find a local formulation of the normalization condition. Let us consider a spacelike hypersurface $ \Sigma_0 $ and a point $ p $ in $ M $ lying on $ \Sigma_0 $. Its neighbourhood in $ \Sigma_0 $ with a characteristic dimension $ \varrho $ will be called $ N_0 $ (it may be e.g. an open ball in chosen coordinates of $ \Sigma_0 $ with the radius $ \varrho $). Now, consider another hypersurface $ \Sigma_1 $ which coincides with $ \Sigma_0 $ on $ \Sigma_0 \smallsetminus N_0 $. Analogically, we denote $ N_1 = \Sigma_1 \smallsetminus (\Sigma_0 \smallsetminus N_0) $. The illustration of the geometry is depicted in Fig. \ref{fig:surfaces}.

\begin{figure}[ht]
\centering
\begin{tikzpicture}[scale=1]
     \tikzstyle{p}=[circle,minimum size=3pt,inner sep=0pt,fill=gray]
	 \tikzstyle{edge1} = [draw,-,gray,line width=0.4pt]
	 \tikzstyle{edge2} = [draw,dashed,gray,line width=0.4pt]

	 \draw [edge2]  (-5,-1.4) -- (3,-1.4) -- (5,2) -- (-3,2) -- cycle;
	 \draw [edge1] plot [smooth cycle, tension=1] coordinates {(-2,0) (0,-0.8) (2,0) (0,0.8)};
	 
	 \draw [edge1] plot [smooth, tension=1.0] coordinates {(-1.854,0.3) (0,1.2) (1.854,0.3) };
	 
	 \draw [edge1] (-2,0) -- (0,0);
	 
	 \node[p] (P) at (0,0) {};
	 \node (C) at (0,0.3) {$ p $};
	 
     \node (D) at (0,1.45) {$ N_1 $};
	 \node (C) at (1,0) {$ N_0 $};
	 \node (A) at (2.6,-1.0) {$ \Sigma_0 $};
	 \node (R) at (-1,0.3) {$ \varrho $};

 \end{tikzpicture}
\vspace{0 mm}
\caption{Diagram of the surfaces in question.}
\label{fig:surfaces}
\end{figure}
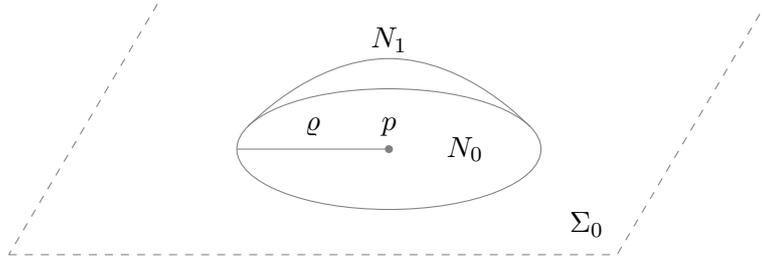

By definition,
\begin{equation}
1 = \int_{\Sigma_0} \int_{\mathrm{u.h.}} A(x,u) ~ u^{\mu} dB_{\mu} ~ u^{\nu} d S_{\nu} = \int_{\Sigma_1} \int_{\mathrm{u.h.}} A(x, u) ~ u^{\mu} dB_{\mu} ~ u^{\nu} d S_{\nu}.
\end{equation}
We split
\begin{equation}
\int_{\Sigma_0} u^{\nu} d S_{\nu} = \int_{\Sigma_0 \smallsetminus N_0} u^{\nu} d S_{\nu} + \int_{N_0} u^{\nu} d S_{\nu}
\end{equation}
(the same for $ \Sigma_1 $) and get
\begin{equation}
\int_{N_0} \int_{\mathrm{u.h.}} A(x,u) ~ u^{\mu} dB_{\mu} ~ u^{\nu} d S_{\nu} = \int_{N_1} \int_{\mathrm{u.h.}} A(x,u) ~ u^{\mu} dB_{\mu} ~ u^{\nu} d S_{\nu}. \label{equal_Ns}
\end{equation}
The region of $ M $ enclosed in between the two hypersurfaces will be called $ \Omega $, and $ \partial \Omega = \bar{N_0} \cup N_1^{-} $, where $ \bar{N_0} $ is the closure of $ N_0 $ and $ N_1^{-} $ is $ N_1 $ with switched orientation. From (\ref{equal_Ns}), it follows
\begin{equation}
\int_{\mathrm{u.h.}} u^{\mu} dB_{\mu} \int_{\partial \Omega} A(x,u) ~ u^{\nu} d S_{\nu} = 0.
\end{equation}
Using Gauss's Theorem (to be found e.g. in \cite{Wald1984}, p. 434) and simplifying, one gets
\begin{equation}
\begin{aligned}
\int_{\partial \Omega} A(x,u) ~ u^{\nu} d S_{\nu} &= \int_{\Omega} \big( A(x, u) u^{\nu} \sqrt{\vert g \vert} \big),_{\nu} ~ d^{4}x = \\ &= \int_{\Omega} \big( A(x, u) u^{\nu} \big);_{\nu} \sqrt{\vert g \vert} ~ d^{4}x = \\ &= \int_{\Omega} \big( A(x, u),_{\nu} u^{\nu} + A (x,u) \Gamma^{\nu}_{~ \nu \varrho} u^{\varrho} \big) \sqrt{\vert g \vert} ~ d^{4}x.
\end{aligned}
\end{equation}
We do not write the $ u^{\nu},_{\nu} $ term, because the velocity is our integration variable, and we do not consider it to be the function of coordinate. Maybe more correctly, we could think of $ u^{\nu} $ as a constant vector field.

So long, we have
\begin{equation}
\int_{\mathrm{u.h.}} u^{\mu} dB_{\mu} \int_{\Omega} \big( A(x, u),_{\nu} u^{\nu} + A (x,u) \Gamma^{\nu}_{~ \nu \varrho} u^{\varrho} \big) \sqrt{\vert g \vert} ~ d^{4}x = 0.
\end{equation}
Now we will proceed with the limit $ \varrho \to 0 $. It produces
\begin{equation}
\int_{\mathrm{u.h.}} \mathrm{vol} \Omega ~ \big( A(x, u),_{\nu} u^{\nu} + A (x,u) \Gamma^{\nu}_{~ \nu \varrho} u^{\varrho} \big) ~ u^{\mu} dB_{\mu} = 0,
\end{equation}
where we denote $ \mathrm{vol} \Omega = \int_{\Omega} \sqrt{\vert g \vert} ~ d^{4}x $. Dividing the equation by the (small, yet nonzero) volume of the shrinking $ \Omega $, we arrive at
\begin{equation}
\int_{\mathrm{u.h.}} \big( A(x, u),_{\nu} u^{\nu} + A (x,u) \Gamma^{\nu}_{~ \nu \varrho} u^{\varrho} \big) ~ u^{\mu} dB_{\mu} = 0. \label{local_normalization}
\end{equation}
For a flat spacetime, we get the intuitive result: the change of the p.p.d. $ A $ in the direction of $ u $, summed over all such directions, is zero. The second term may be interpreted based on the identity $ \Gamma^{\nu}_{~ \nu \varrho} = \ln \sqrt{\vert g \vert},_{\varrho} $ as thinning or thickening of $ A $ due to the expansion or contraction (respectively) of the volume element.

\subsubsection*{Kinematics}
Our next goal is to describe the kinematics of the probabilistic particle. In a classical spacetime, there is a very natural and easy way of treating this problem. We propose the following postulate:
\begin{center}
\textit{The (geometrical) density of the p.p.d. is constant along a classical geodesic.}
\end{center}
Take a geodesic $ \gamma(\tau) $ in a spacetime $ M $. Then our postulate reads
\begin{equation}
\frac{d}{d \tau} \left( \sqrt{\vert g \vert} A(x,u) \right) = 0. \label{kinematical_law}
\end{equation}
Here, the derivative is taken with respect to the proper time along the geodesic, meaning that, implicitly, $ x $ is a point given by $ \gamma(\tau) $ and $ \left. u^{\mu} = \frac{d x^{\mu}}{d \tau} \right\vert_{\gamma(\tau)} $. Taking the partial derivatives,
\begin{equation}
A(x,u) \sqrt{\vert g \vert},_{\mu} u^{\mu} + \sqrt{\vert g \vert} A(x,u),_{\mu} u^{\mu} + \sqrt{\vert g \vert} \frac{\partial A(x,u)}{\partial u^{\mu}}\frac{d u^{\mu}}{d \tau} = 0. \label{kinematical_law_in_parts}
\end{equation}
We note that the function $ A $ is subjected to a constraint of velocity normalization, which will affect the partial derivative w.r.t. velocity. In the coordinate representation, the constraint may (and later will) be implemented by the independence of $ A $ on the time-component of the velocity $ u_t $; then the derivative w.r.t. $ u_t $ is zero.\\

The kinematics is brought into the picture by the geodesic equation,
\begin{equation}
\frac{d u^{\mu}}{d \tau} + \Gamma^{\mu}_{~ \varrho \sigma} u^{\varrho} u^{\sigma} = 0. \label{geodesic}
\end{equation}
Let us assume $ \sqrt{\vert g \vert} \neq 0 $ and divide by it; then we can rewrite the equation (\ref{kinematical_law_in_parts}) using (\ref{geodesic}) in the form (using the coma notation for the partial derivative)
\begin{equation}
A(x,u) \ln \sqrt{\vert g \vert},_{\mu} u^{\mu} + A(x,u),_{\mu} u^{\mu} - A(x,u),_{u^{\mu}} \Gamma^{\mu}_{~ \varrho \sigma} u^{\varrho} u^{\sigma} = 0. \label{pp_kinematics}
\end{equation}
This is the local kinematical law for p.p.d., which is an equivalent of the geodesic equation for classical test particle. The interpretation of the terms is the following: the first term, as seen before, contributes the change of the p.p.d. when the volume element expands or contracts; the second term is a change of the p.p.d. in a given spatial direction; and the last term stands for the change in the velocity space (note that only velocities $ u_x $ from $ T_x M $ are involved, so we need not compare velocities from different tangent spaces). This equation should hold in any point of any geodesic, i.e., all $ x $ and $ u $.\\

Now we shall analyse the relationship between kinematics and normalization. Let (\ref{kinematical_law}) hold in a point $ x $ of $ M $ and again, let $ \gamma(\tau) $ be a geodesic including point $ x $, whose unit tangent vector at $ x $ is $ u $. We will start with the l.h.s. of the local normalization condition (\ref{local_normalization}) and use (\ref{pp_kinematics}),
\begin{equation}
\begin{aligned}
& \int_{\mathrm{u.h.}} \left( A,_{\nu} u^{\nu} + A \Gamma^{\nu}_{~ \nu \varrho} u^{\varrho} \right) ~ dB = \\ = & \int_{\mathrm{u.h.}} \left( A,_{u^{\mu}} \Gamma^{\mu}_{~ \varrho \sigma} u^{\varrho} u^{\sigma} - A \ln \sqrt{\vert g \vert},_{\mu} u^{\mu} + A \Gamma^{\nu}_{~ \nu \varrho} u^{\varrho} \right) ~ dB = \\ = & \int_{\mathrm{u.h.}} A,_{u^{\mu}} \Gamma^{\mu}_{~ \varrho \sigma} u^{\varrho} u^{\sigma} ~ dB.
\end{aligned}
\end{equation}
Now we will go back to the proper time derivative and rewrite it in the form of limit,
\begin{equation}
\begin{aligned}
& \int_{\mathrm{u.h.}} A(x,u),_{u^{\mu}} \Gamma^{\mu}_{~ \varrho \sigma} u^{\varrho} u^{\sigma} ~ dB = - \int_{\mathrm{u.h.}} A(x,u),_{u^{\mu}} \frac{d u^{\mu}}{d \tau} ~ dB = \\ =
& - \lim_{\delta \tau \to 0} \left( \frac{1}{\delta \tau} \int_{\mathrm{u.h.}} \Big( A(x, u_1) - A(x, u_0) \Big) ~ dB_0 \right),
\end{aligned}
\end{equation}
where $ u_1 = u_0 + \delta u $ and $ \delta u = \frac{d u^{\mu}}{d \tau} \delta \tau $. By writing $ d B_0 $, we indicate integrating over $ u_0 $. Next, we will break the integral in two and perform trivial substitution in the first one. It holds
\begin{equation}
\frac{\partial u_0^{\alpha}}{\partial u_1^{\beta}} = \frac{\partial}{\partial u_1^{\beta}} \left( u_1^{\alpha} +  \Gamma^{\alpha}_{~ \varrho \sigma} u_0^{\varrho} u_0^{\sigma} \delta \tau \right) = \delta^{\alpha}_{\beta}.
\end{equation}
We note that it is also possible to write $ u_0^{\alpha} = u_1^{\alpha} + \Gamma^{\alpha}_{~ \varrho \sigma} u_1^{\varrho} u_1^{\sigma} \delta \tau $, in which case we could simply change back, $ u_0^{\alpha} = u_1^{\alpha} + \Gamma^{\alpha}_{~ \varrho \sigma} ( u_0^{\varrho} + \delta u^{\varrho}) ( u_0^{\sigma} + \delta u^{\sigma} ) \delta \tau = u_1^{\alpha} + \Gamma^{\alpha}_{~ \varrho \sigma} u_0^{\varrho} u_0^{\sigma} \delta \tau + O(\delta \tau^{2}) $. Therefore, the Jacobian can produce additional terms of the order $ O(\delta \tau^{2}) $, which we should keep in mind. Performing the substitution, we get
\begin{equation}
\int_{\mathrm{u.h.}} A(x, u_1) ~ dB_0 = \int_{\mathrm{u.h.}} A(x, u_1) ~ dB_1 + O(\delta \tau^{2}).
\end{equation}
This is the same old integral over velocities in point $ x $, we only call the velocities by a new name $ u_1 $. This gives us
\begin{equation}
\begin{aligned}
& - \lim_{\delta \tau \to 0} \left( \frac{1}{\delta \tau} \int_{\mathrm{u.h.}} \Big( A(x, u_1) - A(x, u_0) \Big) ~ dB_0 \right) = \\ =
& - \lim_{\delta \tau \to 0} \left( \frac{1}{\delta \tau} \left( \int_{\mathrm{u.h.}} A(x, u_1) ~ dB_1 + O(\delta \tau^{2}) - \int_{\mathrm{u.h.}} A(x, u_0) ~ dB_0 \right) \right) = \\ =
&  - \lim_{\delta \tau \to 0} \frac{O(\delta \tau^{2})}{\delta \tau} = 0.
\end{aligned}
\end{equation}
To conclude, we have shown that
\begin{equation}
\int_{\mathrm{u.h.}} A(x,u),_{u^{\mu}} \Gamma^{\mu}_{~ \varrho \sigma} u^{\varrho} u^{\sigma} ~ dB = 0.
\end{equation}
Assuming (\ref{pp_kinematics}), this implies (\ref{local_normalization}). This is a satisfactory result, which proves an important property of the postulated kinematics:
\begin{center}
\textit{The kinematical law conserves normalization.}
\end{center}

\subsubsection*{Inhabiting Probabilistic Spacetime}
Up to now, we have been studying the probabilistic particle travelling through a classical spacetime. Now we would like to see how we can use our description to investigate the kinematical properties of a \emph{probabilistic} spacetime $ \mathcal{M} $.

\begin{definition}
Let $ \mathcal{M} = (M, \mathcal{g}) $ be a p.st. A hypersurface $ \Sigma $ is called \emph{exclusive partial Cauchy surface}, if it is a partial Cauchy surface for any realization $ g $ of the probabilistic metric $ \mathcal{g} $.
\end{definition}

We alter Def. \ref{def:ppd} as follows.

\begin{definition} \label{def:ppd_in_pst}
\emph{Probabilistic particle density in a probabilistic spacetime} $ \mathcal{M} = (M,\mathcal{g}) $ is a scalar function $ A(x, U) $ on $ TM $ with values from $ \mathbb{R}_0^{+} $. Let $ \Sigma $ be an exclusive partial Cauchy surface. Then there must be a function of the metric field realization $ N(g) $, such that for every $ \Sigma $, $ A $ satisfies
\begin{equation}
\int_{\Sigma} \int_{\mathrm{u.h.}} A(x, U) ~ u^{\mu} dB_{\mu} ~ u^{\nu} d S_{\nu} = N(g). \label{ppd_in_pst_completeness}
\end{equation}
\end{definition}
The interpretation of $ A(x, U) $ is that it describes the probability density of the particle to be found at point $ x $ with proper speed $ u $, which is a multiple of $ U $. Exact relation of the two velocities shall be clear in a moment.\\

We can normalize to unity in one chosen realization of the metric field; then the integral will have some, generally different, values in other realizations. The definition only states that all the values must be conserved and independent of the hypersurface. Another possibility is to adopt an additional requirement on $ N(g) $, for instance
\begin{equation}
\int N(g) ~ dG = 1,
\end{equation}
which is suggesting that the particle is ``distributed'' in between realizations.\\

Now we take a look at the velocity $ U $ in the argument of the p.p.d. One cannot use the proper velocity $ u $, since it has different meaning under different metric realizations. That is because it is normalized, each time with a different $ g $. Thus, strictly speaking, we should always say ``velocity $ u $ w.r.t. $ g $''. Now, to describe the kinematics of the particle, we would like to study evolution of $ A $ when sliding from, say, point $ x_0 $ to a close point $ x_1 = x_0 + \delta x $ in the future of $ x_0 $. We can specify $ \delta x $ (to first order) by fixing the velocity, i.e., the tangent vector of a geodesic passing through the point $ x_0 $. However, the velocity must be rescaled to fix $ \delta x $ uniquely across all realizations.\\

We progress as follows. We choose a vector $ U $ in the tangent space $ T_{x_0} M $. This vector picks a sheaf of geodesics (w.r.t. different realizations of $ g $) $ \gamma(\tau) $ such that the proper velocity $ u = \frac{d x}{d \tau} \vert_{\gamma} $ is a multiple of $ U $; say, $ U = c u $. We can also write that $ U = \frac{d x}{d y} \vert_{\gamma} $, then it is a tangent of any of the geodesics under a trivial reparametrization $ \tau = c y $. Here, $ c = c(g) $ is a real constant dependent on the metric realization (yet also the velocity!) and $ y $ is the new (universal) parameter of the geodesics. We will later study the evolution of $ A $, for which $ U $ has a universal, metric-independent meaning. One has
\begin{equation}
\delta x^{\mu} = U_0^{\mu} \delta y + O(\delta y^{2}), \label{deltax}
\end{equation}
and the change of the velocity itself is
\begin{equation}
\delta U^{\mu} = - \Gamma^{\mu}_{~ \varrho \sigma} U_0^{\varrho} U_0^{\sigma} \delta y + U_0^{\mu} \frac{d y}{d \tau} \frac{d^{2}\tau}{d y^{2}} \delta y + O(\delta y^{2}), \label{deltaU}
\end{equation}
where we have used the definition of $ U $ and the geodesic equation (\ref{geodesic}). Since our reparametrization is linear, the second term has no effect.\\

The last question we have to answer is what are the numbers $ c $. Since it holds
\begin{equation}
d \tau = \sqrt{- g \left( \frac{d x}{d y}, \frac{d x}{d y} \right) } d y,
\end{equation}
one has simply
\begin{equation}
c = \sqrt{- g \left( U,U \right) }.
\end{equation}

\newpage

\subsubsection*{Kinematics of p.p. in p.st.}
We can eventually consider the problem of kinematics in the p.st. We will see that even for the probabilistic particle, which seems to be a suitable model for such a purpose, the kinematics is still tricky to define.\\

Let us start with the following. We shall assume that the particle represented by its p.p.d. function $ A(x_0, U_0) $ on an exclusive partial Cauchy surface $ \Sigma_0 $ given by $ t_0 = const. $ and evolves independently in all the realizations for some time $ \delta t $, resulting in a collection of functions $ A_g(x_1, U_1) $ only to be composed back, forming the p.p.d. on a successive hypersurface $ \Sigma_1 $. So, in the single step $ \delta t = t_1 - t_0 $, the particle gets to choose the metric realization just \emph{once}. The described procedure results in the following prescription for the p.p.d. in place $ x_1 $ (comprising time $ t_1 $ and hypersurface-position $ x^{i}_1 $) and velocity $ U_1 $:
\begin{equation}
 A(x_1, U_1) = \frac{1}{\int \sqrt{\vert g_1 \vert} ~ dG} \int A_g(x_1, U_1) \sqrt{\vert g_1 \vert} ~ dG, \label{A1_unexpanded}
\end{equation}
where we write $ g_1 $ instead of $ g(x_1) $ for brevity. The relations between the successive positions and velocities are implicitly fixed by $ x_1 - x_0 = \delta x $ and $ U_1 - U_0 = \delta U $, respondent to (\ref{deltax}) and (\ref{deltaU}). Now, using the kinematical law for p.p.d. in classical spacetime in the form
\begin{equation}
A_g(x_1, U_1) \sqrt{\vert g_1 \vert} = A_g(x_0, U_0) \sqrt{\vert g_0 \vert} + O(\delta y^{2})
\end{equation}
yields (in this moment, we have to replace the function $ A(x,U) $ by its coordinate representation $ A(x^{\mu}, U^{\mu}) $ to be able to express the change in the argument)
\begin{equation}
\begin{aligned}
A(x_1^{\mu}, U_1^{i}) & = \frac{1}{\int \sqrt{\vert g_1 \vert} ~ dG} \int A_g(x_0^{\mu}, U_0^{i}) \sqrt{\vert g_0 \vert} ~ dG + O(\delta y^{2}) = \\
& = \frac{1}{\int \sqrt{\vert g_1 \vert} ~ dG} \int A(x_0^{\mu}, U_0^{i}) \sqrt{\vert g_0 \vert} ~ dG + O(\delta y^{2}) = \\
& = \frac{1}{\int \sqrt{\vert g_1 \vert} ~ dG} \int A(t_0, x_1^{i} - U_1^{i} \delta y, U_1^{i} + \Gamma^{i}_{~ \varrho \sigma} U_1^{\varrho} U_1^{\sigma} \delta y) \\ & \qquad \qquad \qquad \sqrt{\vert g(x_1^{\mu} - U_1^{\mu} \delta y) \vert} ~ dG + O(\delta y^{2}). \label{A1}
\end{aligned}
\end{equation}
The equality of the first and the second line of (\ref{A1}) follows from the fact that at $ (x_0^{\mu}, U_0^{i}) $, we begin with a single initial condition $ A_g(x_0^{\mu}, U_0^{i}) = A(x_0^{\mu}, U_0^{i}) $. We note that since it is $ \delta t $ and not $ \delta y $ which is considered to be constant, one should further replace $ \delta y $ by $ \tfrac{\delta t}{U_t} $ everywhere. For now, we will not do it explicitly. In the argument of $ A $, we are writing $ U^{i} $ assuming the coordinate representation of the p.p.d. to be only function of the spatial components of the velocity; the time component $ U^{0} \equiv U_t $ is fixed from normalization of $ U $.\\

Let us expand the right-hand-side to first order in $ \delta y $, or, equivalently, in $ \delta t $, and see whether we can find a prescription for the limiting case of infinitesimal time steps. We will not be expanding in the time component of the position vector, since we are interested in the relation between values of the p.p.d. in different times. It holds
\begin{equation}
\begin{gathered}
A(t_0, x_1^{i} - U_1^{i} \delta y, U_1^{i} + \Gamma^{i}_{~ \varrho \sigma} U_1^{\varrho} U_1^{\sigma} \delta y) = \\ = A(t_0, x_1^{i}, U_1^{i}) + \left( - A,_{i} U^{i} + A,_{U^{i}} \Gamma^{i}_{~ \varrho \sigma} U^{\varrho} U^{\sigma}  \right)_{(t_0,x_1^{i},U_1^{i})} \delta y + O(\delta y^{2}).
\end{gathered}
\end{equation}
Also, (and here we do expand in the time component)
\begin{equation}
\sqrt{\vert g(x_1^{\mu} - U_1^{\mu} \delta y) \vert} = \sqrt{\vert g_1 \vert} -  \sqrt{\vert g_1 \vert},_{\mu} U_1^{\mu} \delta y + O(\delta y^{2}).
\end{equation}

Plugging these into (\ref{A1}), we get
\begin{equation}
\begin{aligned}
A(x_1^{\mu}, U_1^{i}) & = \int \Big( A(t_0, x_1^{i}, U_1^{i}) + \left( - A,_{i} U^{i} + A,_{U^{i}} \Gamma^{i}_{~ \varrho \sigma} U^{\varrho} U^{\sigma}  \right)_{(t_0,x_1^{i},U_1^{i})} \delta y \Big) \\ & \qquad \qquad \frac{\sqrt{\vert g_1 \vert} -  \sqrt{\vert g_1 \vert},_{\mu} U_1^{\mu} \delta y}{\int \sqrt{\vert g_1 \vert} ~ dG} ~ dG + O(\delta y^{2}) = \\
& = \int \Bigg( A(t_0, x_1^{i}, U_1^{i}) \frac{\sqrt{\vert g_1 \vert} -  \sqrt{\vert g_1 \vert},_{\mu} U_1^{\mu} \delta y}{\int \sqrt{\vert g_1 \vert} ~ dG} + \\ &  ~ ~ ~ \left( - A,_{i} U^{i} + A,_{U^{i}} \Gamma^{i}_{~ \varrho \sigma} U^{\varrho} U^{\sigma}  \right)_{(t_0,x_1^{i},U_1^{i})} \delta y \frac{\sqrt{\vert g_1 \vert}}{\int \sqrt{\vert g_1 \vert} ~ dG} \Bigg) ~ dG + O(\delta y^{2}) = \\
& = A(t_0, x_1^{i}, U_1^{i}) \Big( 1 - \frac{\int \sqrt{\vert g_1 \vert},_{\mu} dG}{\int \sqrt{\vert g_1 \vert} ~ dG} U_1^{\mu} \delta y \Big) + \\ &  ~ ~ ~ \Big( - A,_{i} U^{i} + A,_{U^{i}} \frac{\int \Gamma^{i}_{~ \varrho \sigma} \sqrt{\vert g_1 \vert} ~ dG}{\int \sqrt{\vert g_1 \vert} ~ dG}  U^{\varrho} U^{\sigma} \Big)_{(t_0,x_1^{i},U_1^{i})} \delta y + O(\delta y^{2}). \label{AAAA}
\end{aligned}
\end{equation}

Subtracting values of $ A $ on different time slices, dividing by $ \delta y $ and letting $ \delta y \to 0 $ eventually produces (we are dropping the no-longer-necessary indices)
\begin{equation}
A,_t U_t = - A ~ \frac{\int \sqrt{\vert g \vert},_{\mu} dG}{\int \sqrt{\vert g \vert} ~ dG} U^{\mu} - A,_{i} U^{i} + A,_{U^{i}} \frac{\int \Gamma^{i}_{~ \varrho \sigma} \sqrt{\vert g \vert} ~ dG}{\int \sqrt{\vert g \vert} ~ dG}  U^{\varrho} U^{\sigma}. \label{eq:AAA}
\end{equation}
The first term after the equal sign in (\ref{eq:AAA}) comes from the second term in the bracket at line 5 of (\ref{AAAA}), the last two terms in (\ref{eq:AAA}) come from the two terms in the bracket at line 6 of (\ref{AAAA}).\\

In yet another words,
\begin{equation}
\int \Big( A \sqrt{\vert g \vert},_{\mu} U^{\mu} +  A,_{\mu} U^{\mu} \sqrt{\vert g \vert} - A,_{U^{i}} \Gamma^{i}_{~ \varrho \sigma} U^{\varrho} U^{\sigma} \sqrt{\vert g \vert} \Big) ~ dG = 0. \label{pp_pst_kinematics}
\end{equation}
This is a differential equation for the p.p.d., if the particle is allowed to choose a metric realization at each time and the respective versions of the p.p.d. obeying different realizations are summed instantly. We see that it is an analogue of the kinematical law for the p.p.d. in classical spacetime, cf. (\ref{pp_kinematics}). Disappointingly, the inevitable result of our construction is that this prescription gives mere evolution along \emph{mean geodesics}, cf. (\ref{mean_geodesic}). Thus, the probabilistic particle subjected to this kinematical law is not responding to the p.st. in any truly interesting way.\\

Let us hereby comment on the used approach. We were interested in ways to study the kinematical properties of a p.st. and came up with the probabilistic particle, which was intended to evolve w.r.t. the probabilistic metric in a way that would depict its geometrical content. However, the law that we derived intuitively does not give a better image than the notion of mean geodesic. This is, interestingly, because of the limiting procedure $ \delta y \to 0 $, which results in averaging the effect of the probabilistic geometry. If the steps are infinitesimal, the particle is not subjected to any dissolving; it behaves as if there was some mean metric. The analogy here would be with Brownian motion, where one also needs \emph{finite} steps to observe dispersion in positions of the particle.\\

If we want to observe some nontrivial behaviour of the probabilistic particle, there is a possibility of using (\ref{A1_unexpanded}) with \emph{finite} steps $ \delta t $ between the time slices. Of course, such a prescription does not have a good physical justification; and moreover, it takes an arbitrary coordinate-dependent parameter $ \delta t $, and therefore is not covariant; however, it may ultimately meet the goal of depicting the behaviour of the probabilistic metric.\\

Let us note that maybe the perfectly correct kinematical law for the probabilistic particle would involve letting the particle choose a realization at every instant but summing up only after some finite time; that is, considering all possible paths in the space of realizations and summing over these paths to get the p.p.d. on some later time slice. This approach, employing a path integral, would already resemble a quantum-mechanical computation. The technical demands for implementation, however, would be extremely high.

\section{Kinematical Examples}
In this section, we aim to demonstrate the functioning of the established formalism; and we shall do so by studying the kinematics of a probabilistic particle. We start in a couple of well-known classical spacetimes and work our way towards a p.st.

\subsection{Minkowski Spacetime}
We start with a simple case of Minkowski spacetime. For both simplicity and later use, we will be interested in a particle with spherically symmetric distribution of probability density and only radial velocity. That is why we choose to work in spherical coordinates, in which the metric has the form
\begin{equation}
ds^{2} = - dt^{2} + dr^{2} + r^{2} d \Omega^{2}, \label{Mink_spher}
\end{equation}
where $ d \Omega^{2} = d \theta^{2} + \sin^{2} \theta d \varphi^{2}  $ is the angular part. Our test particle is described by a spherically symmetric p.p.d. $ A(x,u) $ (a function on the tangent bundle), or $ A(t,r,u_r) $ (its coordinate representation). The particle is assumed to be moving only radially, e.g., the angular parts of velocity $ u^{\mu} = \left( u_{t}, u_{r},0,0 \right) $ (we write lower indices for aesthetic reasons) are zero. Furthermore, there is no dependence on $ u_t $, since we fix this component from the normalization of velocity,
\begin{equation}
u_t = \sqrt{u_r^{2} + 1}.
\end{equation}
From now on, whenever we write $ u_t $, we will implicitly mean $ u_t(u_r) $.\\

We should provide normalization for the p.p.d. Our partial Cauchy surface (or in this case, Cauchy surface) is $ t = const. $ We can use (\ref{vel_int}), however, we have to adapt it for the spherical coordinates in the space of velocities, which can be done by replacing $ d^{4}u $ by $ du_t ~ r^{2} \sin \theta ~ d u_r  d u_{\theta}  d u_{\varphi} $ (luckily, our change of coordinates does not interfere with the u.h. constraint since it involves only the spatial part, otherwise we would have to repeat the calculation for the new coordinates). We have
\begin{equation}
\begin{aligned}
& \int_{\Sigma} \int_{\text{u.h.}} A(x,u) u^{\mu} dB_{\mu} u^{\nu} d S_{\nu} = \\
= & \int_{\Sigma} \int_{\text{u.h.}} - A(t,r,u_r) ~ u^{\nu} n_{\nu} \sqrt{\vert h \vert} r^{2} \sin \theta ~ d u_t  d u_r  d u_{\theta}  d u_{\varphi} ~  dr  d \theta  d \varphi = \\
= & \int_{\Sigma} \int_{\text{u.h.}} A(t,r,u_r) ~ u_t (r^{2} \sin \theta)^{2} ~ d u_r  d u_{\theta}  d u_{\varphi} ~ dr  d \theta  d \varphi = \\
= & \int_{\Sigma} \int_{\text{u.h.}} A(t,r,u_r) ~ u_t ~ r^{4} ~ 2 \pi^{4} ~ d u_r dr. \label{Mink_normalization}
\end{aligned}
\end{equation}

Let us look at the kinematics of the probabilistic test particle. It is given by the equation (\ref{pp_kinematics}). In our case, the volume element is
\begin{equation}
\sqrt{\vert g \vert} = r^{2} \sin \theta
\end{equation}
and the relevant Christoffel symbols $ \Gamma^{r}_{~ tt}, \Gamma^{r}_{~ tr}, \Gamma^{r}_{~ rt}, \Gamma^{r}_{~ rr} $ are all zero. The equation reads
\begin{equation}
A 2 r \sin \theta u_r + \left( A,_t u_t +  A,_r u_r \right) r^{2} \sin \theta = 0. \label{Mink_kinematics}
\end{equation}
For $ r^{2} \sin \theta = 0 $ (which describes the coordinate singularities), it holds automatically; in all other cases we can divide by this factor, getting
\begin{equation}
A,_r + A,_t \frac{u_t}{u_r} = -\frac{2}{r} A.
\end{equation}
Using the handbook \cite{Polyanin2001}, we arrive at a solution
\begin{equation}
A = \frac{1}{r^{2}} ~ \phi \left( \frac{u_t}{u_r} ~ r - t \right),
\end{equation}
where $ \phi $ is an arbitrary function. Instead of $ \frac{u_t}{u_r} $, we could also write $ \sqrt{1 + \frac{1}{u_r^{2}}} $. The initial condition for $ t = 0 $ has to be of the form
\begin{equation}
A_{\text{in}} = \frac{1}{r^{2}} ~ \phi \left( \frac{u_t}{u_r} ~ r \right), \label{Mink_solution}
\end{equation}
which, due to the freedom in choosing $ \phi $, indeed covers all possible functions $ A_{\text{in}} $. We note that $ u_r $ is a parameter and as such may be inbuilt within $ \phi $.\\

Let us choose e.g. the Gaussian profile in $ r $ and $ u_r $,
\begin{equation}
A_{\text{in}} = \frac{1}{r^{2}} ~ N \exp \left( - \frac{\left( r - r_{\text{in}} \right)^{2}}{\varepsilon_r} \right) \exp \left( - \frac{\left( u_r - u_{r\text{in}} \right)^{2}}{\varepsilon_u} \right).
\end{equation}
Then the solution is
\begin{equation}
A = \frac{1}{r^{2}} ~ N \exp \left( - \frac{\left( r - \tfrac{u_r}{u_t} ~ t - r_{\text{in}} \right)^{2}}{\varepsilon_r} \right) \exp \left( - \frac{\left( u_r - u_{r\text{in}} \right)^{2}}{\varepsilon_u} \right).
\end{equation}

The constant $ N $ should be fixed from normalization. The integration would not be very nice for our solution; however, we can simplify considerably by an infinitesimal ``sharpening'' of its values. Let us take
\begin{equation}
N = \tilde{N} \frac{1}{\sqrt{\pi \varepsilon_r}} \frac{1}{\sqrt{\pi \varepsilon_u}}.
\end{equation}
Then, in the limit $ \varepsilon_r \to 0^{+} $ and $ \varepsilon_u \to 0^{+} $, our solution becomes
\begin{equation}
A = \tilde{N} ~ \delta \left( r -  \tfrac{u_r}{u_t} ~ t - r_{\text{in}} \right) ~ \delta (u_r - u_{r\text{in}} ) \label{pp_Mink_sharp}
\end{equation}
in the sense of distributions. Now, the integration is trivial:
\begin{equation}
\int_{\mathbb{R}^{+}_0} \int_{\mathbb{R}} \tilde{N} ~ \delta \left( r -  \tfrac{u_r}{u_t} ~ t - r_{\text{in}} \right) ~ \delta (u_r - u_{r\text{in}} ) ~ u_t ~ r^{4} ~ 2 \pi^{4} ~ d u_r dr = \tilde{N} u_{t \text{in}} r_{\text{in}}^{2} 2 \pi^{4}.
\end{equation}
Equating this result to 1, we can fix the constant,
\begin{equation}
\tilde{N} = \frac{1}{u_{t \text{in}} r_{\text{in}}^{2} 2 \pi^{4}}.
\end{equation}
We see, at least in this infinitesimal case, that the normalization is independent of time, as has been expected. Also, the form (\ref{pp_Mink_sharp}) is consistent with kinematics of the classical test particle, since the radial velocity is constant and
\begin{equation}
r =  r_{\text{in}} + u_r \tau, \qquad t = t_{\text{in}} + u_t \tau,
\end{equation}
so, fixing $ t_{\text{in}} = 0 $, it indeed holds
\begin{equation}
r = r_{\text{in}} + \tfrac{u_r}{u_t} ~ t .
\end{equation}

\subsection{Schwarzschild Spacetime}
The second example will be from the Schwarzschild spacetime, where we can progress similarly as in the preceding case. We start with the metric
\begin{equation}
ds^{2} = - \left( 1 - \tfrac{2 M}{r} \right) dt^{2} + \left( 1 - \tfrac{2 M}{r} \right)^{-1} dr^{2} + r^{2} d \Omega^{2}. \label{Schw}
\end{equation}
Again, we will only consider radially infalling spherically symmetric p.p.d. \break $ A(t,r,u_r) $. The velocity of the particle is $ u^{\mu} = \left( u_{t}, u_{r},0,0 \right) $ (we again write lower indices only for aesthetic reasons). The normalization of velocity is
\begin{equation}
u_t = \sqrt{\frac{u_r^{2} + \left( 1 - \tfrac{2 M}{r} \right)}{\left( 1 - \tfrac{2 M}{r} \right)^{2}}}. \label{u_norm}
\end{equation}

We shall provide normalization for the p.p.d. Let us assume that no part of the particle is under the horizon, meaning that $ A = 0 $ for all $ r \leq 2M $. Then we can study the normalization on a hypersurface $ \Sigma $ given by $ t = const. $, because it is acausal at least where the particle is passing through (otherwise there would be a problem with the sub-horizon dynamic region). The unit normal to $ \Sigma $ is
\begin{equation}
n^{\mu} = \left( (1 - \tfrac{2 M}{r} )^{-1/2}, 0,0,0 \right),
\end{equation}
and the induced 3D metric on $ \Sigma $ has the form
\begin{equation}
h_{a b} = \begin{pmatrix}
\left( 1 - \tfrac{2 M}{r} \right)^{-1} & & \\
& r^{2} & \\
& & r^{2} \sin^{2} \theta
\end{pmatrix},
\end{equation}
so $ \sqrt{\vert h \vert} = \left( 1 - \tfrac{2 M}{r} \right)^{-1/2} r^{2} \sin \theta $.\\

Projection of the velocity to the unit normal is
\begin{equation}
u^{\mu} n_{\mu} = - u_t  \left( 1 - \tfrac{2 M}{r} \right)^{1/2}.
\end{equation}

The normalization integration will be similar to that done before. However, we are not in Minkowski spacetime anymore, so to express the velocity integral, we need to find a local basis in which the metric is locally the same as the metric of Minkowski spacetime (in our coordinates). The transformation is of the form
\begin{equation}
g_{\text{S} \alpha \beta}(x) = \Lambda^{\mu}_{~ \alpha}(x) \Lambda^{\nu}_{~ \beta}(x) g_{\text{M} \mu \nu}
\end{equation}
and the solution for the matrices is
\begin{equation}
\Lambda^{0}_{~ 0} = \left( 1 - \tfrac{2 M}{r} \right)^{1/2}, \qquad \Lambda^{1}_{~ 1} = \left( 1 - \tfrac{2 M}{r} \right)^{-1/2}, \qquad \text{others are zero.}
\end{equation}
This is nice, since the determinant $ \vert \Lambda \vert = 1 $.\\

Now we can compute the integral. We define a new function of the coordinate values $ \tilde{A} $ via $ A(t,r,u_t,u_r) = \tilde{A}(t,r,\tilde{u}_t,\tilde{u}_r) $, where $ \tilde{u}_t,\tilde{u}_r $ are the components in the coordinate basis which locally gives the metric the form of the metric in Minkowski spacetime, $ \tilde{u}^{\mu} = \Lambda^{\mu}_{~ \alpha} u^{\alpha}$. Then, using the prescription (\ref{vel_int}) as before, we write
\begin{equation}
\begin{aligned}
& \int_{\Sigma} \int_{\text{u.h.}} A(x,u) u^{\mu} dB_{\mu} u^{\nu} d S_{\nu} = \\
& = \int_{\Sigma} \int_{\text{u.h.}} - \tilde{A}(t,r,\tilde{u}_t,\tilde{u}_r) ~ \delta \left( \tilde{u}_t - \tilde{u}_t(r, \tilde{u}_r) \right) ~ u^{\nu} n_{\nu} \sqrt{\vert h \vert} \\ & \qquad \qquad \qquad \qquad \qquad r^{2} \sin \theta ~ d u_t  d u_r  d u_{\theta}  d u_{\varphi} ~  dr  d \theta  d \varphi = \\
& = \int_{\Sigma} \int_{\text{u.h.}} - A(t,r,u_t,u_r) ~ \delta \left( u_t - u_t(r,u_r) \right) \vert 1 - \tfrac{2 M}{r}  \vert^{-1/2} ~ u^{\nu} n_{\nu} \sqrt{\vert h \vert} \\ & \qquad \qquad \qquad \qquad \qquad r^{2} \sin \theta ~ d u_t  d u_r  d u_{\theta}  d u_{\varphi} ~ dr  d \theta  d \varphi = \\
& = \int_{\Sigma} \int_{\text{u.h.}} - A(t,r,u_t,u_r) ~ \vert 1 - \tfrac{2 M}{r}  \vert^{-1/2} ~ u^{\nu} n_{\nu} \sqrt{\vert h \vert} \\ & \qquad \qquad \qquad \qquad \qquad r^{2} \sin \theta ~  d u_r  d u_{\theta}  d u_{\varphi} ~ dr  d \theta  d \varphi = \\
& = \int_{\Sigma} \int_{\text{u.h.}} - A(t,r,u_t,u_r) ~ \vert 1 - \tfrac{2 M}{r}  \vert^{-1/2} (-) u_t  \left( 1 - \tfrac{2 M}{r} \right)^{1/2} \\ & \qquad \qquad \left( 1 - \tfrac{2 M}{r} \right)^{-1/2} r^{2} \sin \theta ~ r^{2} \sin \theta ~  d u_r  d u_{\theta}  d u_{\varphi} ~ dr  d \theta  d \varphi = \\
& = \int_{\Sigma} \int_{\text{u.h.}} A(t,r,u_t,u_r) ~ \vert 1 - \tfrac{2 M}{r}  \vert^{-1/2} ~ u_t  \\ & \qquad \qquad \qquad \qquad \qquad (r^{2} \sin \theta)^{2} ~ d u_r  d u_{\theta}  d u_{\varphi} ~ dr  d \theta  d \varphi = \\
& = \int_{\Sigma} \int_{\text{u.h.}} A(t,r,u_t,u_r) ~ \vert 1 - \tfrac{2 M}{r}  \vert^{-1/2} ~ u_t ~ r^{4} ~ 2 \pi^{4} ~ d u_r dr. \label{Schw_normalization}
\end{aligned}
\end{equation}

Once again, let us choose a sharp value of the initial radius and radial velocity at $ t = 0 $,
\begin{equation}
A(0,r,u_t,u_r) \equiv A_{\text{in}} = \tilde{N} \delta(r - r_{\text{in}}) \delta(u_r - u_{r \text{in}}). \label{A0_choice}
\end{equation}
Then we can evaluate the integral explicitly,
\begin{equation}
\int_{\Sigma} \int_{\text{u.h.}} A(x,u) u^{\mu} dB_{\mu} u^{\nu} d S_{\nu} = 2 \pi^{4} \tilde{N} u_{t \text{in}} r_{\text{in}}^{4} \left( 1 - \tfrac{2 M}{r_{\text{in}}} \right)^{-1/2},
\end{equation}
where $ u_{t \text{in}} = u_{t \text{in}}(r_{\text{in}}, u_{r \text{in}}) $ comes from (\ref{u_norm}) and we do not write the absolute value since we assume to be above horizon ($ r_{\text{in}} > 2M $) anyway. Let us normalize the p.p.d to 1. Then, for our choice (\ref{A0_choice}), we are getting the value of the constant,
\begin{equation}
\tilde{N} = \frac{ \left( 1 - \tfrac{2 M}{r_{\text{in}}} \right)^{1/2}}{2 \pi^{4} u_{t \text{in}} r_{\text{in}}^{4} }
\end{equation}
If the particle is static in the beginning, $ u_{r \text{in}} = 0 $, this simplifies to
\begin{equation}
\tilde{N} = \frac{\left( 1 - \tfrac{2 M}{r_{\text{in}} } \right)}{2 \pi^{4} r_{\text{in}}^{4} }.
\end{equation}

Let us write down the kinematical law and see what we can infer from it. We start with the equation (\ref{pp_kinematics}). We will need the volume element, which is
\begin{equation}
\sqrt{\vert g \vert} = r^{2} \sin \theta.
\end{equation}
The Christoffel symbols of interest (concerning our symmetry) are
\begin{equation}
\Gamma^{r}_{~ tt} = \left( 1 - \tfrac{2M}{r} \right) \tfrac{M}{r^{2}}, \qquad \Gamma^{r}_{~ tr} = \Gamma^{r}_{~ rt} = 0, \qquad \Gamma^{r}_{~ rr} = - \left( 1 - \tfrac{2M}{r} \right)^{-1} \tfrac{M}{r^{2}}. \label{Christoffels}
\end{equation}
Plugging these into (\ref{pp_kinematics}), one obtains
\begin{equation}
\begin{aligned}
& A 2 r \sin \theta u_r + \left( A,_t u_t + A,_r u_r \right) r^2 \sin \theta \\ & - A,_{u_r} r^2 \sin \theta ~ \frac{M}{r^{2}} \left( \left( 1 - \tfrac{2 M}{r} \right) u_t^{2} - \left( 1 - \tfrac{2 M}{r} \right)^{-1} u_r^{2} \right) = 0. \label{Schw_kinematics_orig}
\end{aligned}
\end{equation}
For $ \sin \theta = 0 $, the equation holds. In all other cases, we divide by this factor. Further, thanks to (\ref{u_norm}), it holds
\begin{equation}
\left( \left( 1 - \tfrac{2 M}{r} \right) u_t^{2} - \left( 1 - \tfrac{2 M}{r} \right)^{-1} u_r^{2} \right) = 1,
\end{equation}
which grants a welcomed simplification. We thus have
\begin{equation}
A 2 r u_r + \left( A,_t u_t + A,_r u_r \right) r^2 - A,_{u_r} M = 0. \label{Schw_kinematics_simple}
\end{equation}
One can find a solution to \ref{Schw_kinematics_simple} just as we did\footnote{We assumed $ A = e^{t Q(r,u_r)}\phi(r,u_r) $ and obtained the forms of the functions $ Q = \psi(\eta - \xi) $ and $ \phi = e^{I(\xi,\eta - \xi)\psi(\eta - \xi)} \varphi(\eta - \xi) $, where $ \eta = \frac{1}{r} $ and $ \xi = \frac{u_r^{2}}{2M} $. The functions $ \psi $ and $ \varphi $ are arbitrary, whereas
\begin{equation}
\begin{aligned}
I(\xi,\eta - \xi) = 2 \ln \xi - 4M \sqrt{\eta - \xi + \tfrac{1}{2M}} \Bigg[ & \frac{\sqrt{\eta}}{2M \xi (\eta-\xi)}-\frac{2 \tanh^{-1}\left(\frac{\sqrt{\eta}}{\sqrt{\eta - \xi + \frac{1}{2M}}}\right)}{\sqrt{\eta - \xi + \frac{1}{2M}}} + \\ & + \frac{(2 \eta- 2 \xi - \frac{1}{2M}) \tanh ^{-1}\left(\frac{\sqrt{\eta}}{\sqrt{\eta - \xi}}\right)}{(\eta - \xi)^{3/2}} \Bigg].
\end{aligned}
\end{equation}
This is just one special case, and we suspect that it cannot be normalized.}, but it is less easy to come up with a solution for a reasonable (most importantly, normalizable) initial condition. Nevertheless, from the conceptual point of view, such a solution would not be much more interesting than the preceding solution for Minkowski spacetime, since it would again represent some superposition of classical geodesic motions.

\subsection{Spherical Shell}
Now we join the preceding two cases and study the probabilistic particle around a thin massive spherical shell; a setting which will be convenient for our line of thought. We suppose that the shell is of constant radius and lies above its own horizon, so the spacetime is everywhere static. The probabilistic particle will be attracted towards the shell only until it enters the inner flat region. We assume that the test particle does not interact in any way with the shell. Since we again consider only radial motion, i.e. the motion normal to the surface of the shell, the metric induced on the shell will have no effect on the kinematics of the test particle. Effectively, the Christoffel symbols only switch from one set of values to another.\\

Let us denote the coordinates in the inner flat region, whose metric is (\ref{Mink_spher}), by $ T,R $ to differentiate them from the coordinates of the outer Schwarzschild region with the metric (\ref{Schw}). Let $ R = \varrho $ be the radius of the shell, then, so as to have the metric continuous, we shall put $ r = \varrho $ on the shell as well. Both metrics are static, so the time coordinates can be shifted arbitrarily, and we choose $ t = T $ on the shell.\\

For normalization, we can choose the hypersurface $ \Sigma $ given by $ t = T = const. $ Thanks to the equality of time coordinates on the shell, $ \Sigma $ is well defined by this choice. The normalization integral is given by (\ref{Mink_normalization}) for the inner and by (\ref{Schw_normalization}) for the outer region.\\

The kinematic equations are also known to us from the preceding considerations; they are given by (\ref{Mink_kinematics}) for the inner and by (\ref{Schw_kinematics_simple}) for the outer region. We point out one detail which has not been important so far but will be necessary here: the radially infalling particle will, at some point, reach $ r = 0 $, and it should pass the origin and continue its movement on the other side. To cover such scenario by our description, we need to flip the sign of $ u_R $ and $ u_r $ once the particle passes through the origin. This can be arranged by means of analytic continuation to include $ r < 0 $ and $ R < 0 $. The equation for inner p.p.d. in this continuation stays unchanged,
\begin{equation}
A^{\mathrm{in}} 2 R u_R + \left( A^{\mathrm{in}},_T u_T +  A^{\mathrm{in}},_R u_R \right) R^{2} = 0, \label{Mink_kinematics_continued}
\end{equation}
but in the equation for outer p.p.d., we have to alter the sign of the attractive term, reading
\begin{equation}
A^{\mathrm{out}} 2 r u_r + \left( A^{\mathrm{out}},_t u_t + A^{\mathrm{out}},_r u_r \right) r^2 - \text{sgn}(r) A^{\mathrm{out}},_{u_r} M = 0. \label{RSS_kinematics_continued}
\end{equation}

Now, the setting that we chose for this paragraph is designated to allow for a \emph{static} solution, meaning $ A^{\mathrm{in}},_T = A^{\mathrm{out}},_t = 0 $. There are two reasons why one may be concerned about such a solution: first, it represents an interesting state of the particle which does not have an analogy in the classical case (where the particle in the considered scenario simply cannot be static) and remotely resembles the state of a quantum particle; second, the equations simplify considerably when the time derivatives are set to zero. This could not be done in either of the preceding cases, since the static solutions there are either trivial or non-normalizable.\\

In the static case, the inner solution (\ref{Mink_solution}) simplifies to
\begin{equation}
A^{\mathrm{in}} = \frac{C(u_R)}{R^{2}}.
\end{equation}
The outer solution is yet to be found. With the help of \cite{Polyanin2001}, one easily gets
\begin{equation}
A^{\mathrm{out}} = \frac{\phi(x)}{r^{2}}, \qquad \text{where } x = u_r^{2} - \tfrac{2M}{\vert r \vert},
\end{equation}
where $ \phi $ is an arbitrary function. It is worth noticing that $ x $, if multiplied by half of the mass of the test particle $ m/2 $, resembles its total energy. It would indeed seem reasonable if the static solution was (letting aside the geometrical factor $ r^{-2} $) a function of energy, which is conserved.\\

At first sight, the outer solution is quite worrying: since $ x \to u_r^{2} $ as $ r \to \pm \infty $, the p.p.d. decreases only as fast as $ r^{-2} $ and the normalization integral diverges badly. However, this problem can be resolved if we arrange for $ \phi(x) $ to become zero in some \emph{finite} value of the radius $ r $. Well, let us take for example
\begin{equation}
\phi(x) = D(x) ~ \chi_{[-d,d]} \left( x + \tfrac{2M}{r_{\mathrm{in}}} \right), \label{phi_RS}
\end{equation}
where $ D(x) $ is an arbitrary function and $ \chi_Y(y) $ is the characteristic function of the interval $ Y $ (giving 1 for $ y \in Y $ and 0 otherwise). We assume $ d $ and $ r_{\mathrm{in}} $ to be positive real numbers. Then the characteristic function is non-zero when
\begin{equation}
u_r^{2} - \tfrac{2M}{\vert r \vert} + \tfrac{2M}{r_{\mathrm{in}}} \leq d \qquad \text{and} \qquad u_r^{2} - \tfrac{2M}{\vert r \vert} + \tfrac{2M}{r_{\mathrm{in}}} \geq - d.
\end{equation}
From the first inequality it follows that the greatest value of $ \vert r \vert $, for which $ \phi(x) $ might be non-vanishing, is when $ u_r = 0 $ and
\begin{equation}
d = \tfrac{2M}{\vert r \vert} - \tfrac{2M}{r_{\mathrm{in}}}.
\end{equation}
All greater radii result in the p.p.d. vanishing, which rules out the divergence of the normalization integral at $ r \to \pm \infty $. The second inequality is not so restrictive; it may be viewed as a condition on $ u_r $.\\

For simplicity, let us take $ D = const. $ Having the inner and the outer solution, we must fix the rest of the freedom so that the solution is continuous on the shell. That means
\begin{equation}
C(u_R) = D ~ \chi_{[-d,d]} \left( u_r^{2} - \tfrac{2M}{\varrho} + \tfrac{2M}{r_{\mathrm{in}}} \right).
\end{equation}
Before ensuring this equality, we have to determine what happens with the velocity once the particle enters the inner region. Note that the velocity \emph{must change}, since its normalization changes due to the change in the metric. We argue that it is exactly just the normalization which changes, i.e., the inner velocity (as a vector) is rescaled to become a multiple of the outer velocity after the particle passes through the shell. In fact, one cannot think of any other covariant way to relate the two. So, all we have to do is re-normalize the velocity, which results in
\begin{equation}
u_R = \frac{a u_r}{\sqrt{u_r^{2} + b}}, \qquad \text{where } a = \frac{\left( 1 - \tfrac{2M}{\varrho} \right)}{\sqrt{1 - \left( 1 - \tfrac{2M}{\varrho} \right)^{2}}}, ~ b = \frac{\left( 1 - \tfrac{2M}{\varrho} \right)}{1 - \left( 1 - \tfrac{2M}{\varrho} \right)^{2}}.
\end{equation}
From here, we can express $ C(u_R) $ as a sum of characteristic functions (the function $ u_R(u_r) $ is monotonous, but we need two intervals since $ u_r $ can be either positive or negative). Denoting
\begin{equation}
\Delta_{\overset{1}{0}} = \frac{a \sqrt{\pm d + \tfrac{2M}{\varrho} - \tfrac{2M}{r_{\mathrm{in}}}}}{\sqrt{\pm d + \tfrac{2M}{\varrho} - \tfrac{2M}{r_{\mathrm{in}}} + b}},
\end{equation}
we can write
\begin{equation}
C(u_R) = D \left( \chi_{[\Delta_0,\Delta_1]} ( u_R ) + \chi_{[-\Delta_1,-\Delta_0]} ( u_R ) \right).
\end{equation}
By this, we have provided a full solution up to a normalization constant $ D $.\\

As usual, let us simplify the p.p.d. by a limiting procedure. If we take
\begin{equation}
D = \frac{\tilde{D}}{2 d}
\end{equation}
and let $ d \to 0 $, we get a Dirac delta representation. In the outer region,
\begin{equation}
A^{\mathrm{out}} = \frac{\tilde{D}}{r^{2}} ~ \delta \left( u_r^{2} - \tfrac{2M}{\vert r \vert} + \tfrac{2M}{r_{\mathrm{in}}} \right).
\end{equation}
In the inner region, the solution simplifies to
\begin{equation}
A^{\mathrm{in}} = \frac{\tilde{D}}{2 R^{2}} \frac{ab}{\sqrt{\tfrac{2M}{\varrho} - \tfrac{2M}{r_{\mathrm{in}}}} \left( \tfrac{2M}{\varrho} - \tfrac{2M}{r_{\mathrm{in}}} + b \right)^{3/2} } \Big( \delta (u_R - \Delta_{01}) + \delta (u_R + \Delta_{01}) \Big),
\end{equation}
where $ \Delta_{01} $ is $ \Delta_{0 \text{ or } 1} $ in which one puts $ d = 0 $.

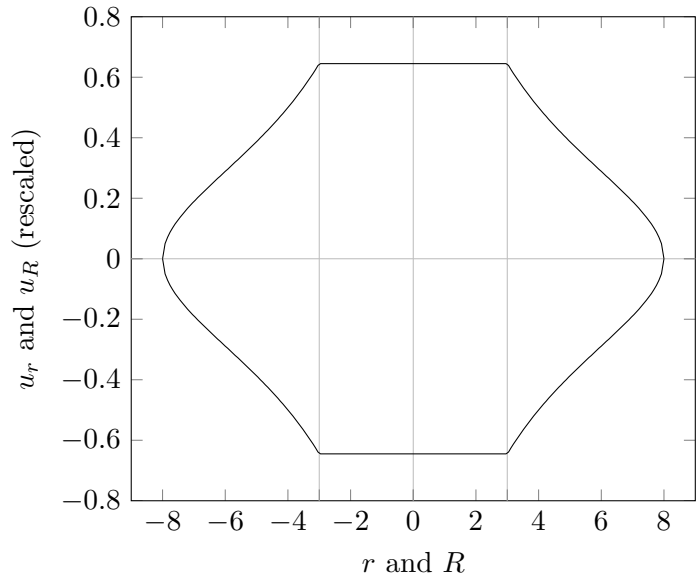
\begin{figure}[ht]
\centering
\begin{tikzpicture}
  \begin{axis}[
	width=9 cm,
	height=8 cm,
    xmin=-9,xmax=9,
    ymin=-0.8,ymax=0.8,
    xlabel=$ r \text{ and } R $,
    ylabel={$ u_r \text{ and } u_R \text{ (rescaled)} $},
    ylabel near ticks,
    xlabel near ticks,
    	/pgf/number format/.cd,
    	1000 sep={},
    extra y ticks       = 0,
    extra y tick labels = ,
    extra y tick style  = { grid = major },
    extra x ticks       = {-3,0,3},
    extra x tick labels = ,
    extra x tick style  = { grid = major },
    ]
  
    
  \addplot[black] table {plots/RSSdelta1.txt};
  \addplot[black] table {plots/RSSdelta1neg.txt};
 
  \end{axis}
\end{tikzpicture}
\vspace{-2 mm}
\caption{The $ \delta $-limit of the p.p.d. plotted as $ u_r(r) $ and $ u_R(R) $ in the outer and the inner region, respectively. We have rescaled $ u_R $ so the values appear continuous, in fact $ \vert u_R \vert $ is smaller due to a different normalization. We have set $ M = 1 $, $ \varrho = 3 $, $ r_{\mathrm{in}} = 8 $. }
\label{fig:shell}
\end{figure}

The non-zero values of the p.p.d. in the radius-velocity plane are plotted in Fig.~\ref{fig:shell}. We obtained a full orbit, or a closed path of a classical test particle which falls towards and through the shell, leaves the interior on the other side, reaches the turning point and returns back to the initial position. The probabilistic particle, however, is everywhere on the orbit, at all times, with the same probability density, so the p.p.d is static.\\

Now it would be the time to fix the constant $ \tilde{D} $, which can be expressed by means of $ M, \varrho $ and $ r_{\text{in}} $. The calculation is similar to that in Minkowski spacetime, with the difference that there is one Dirac delta less now, so the integration is more complicated. Since the result is not in any way interesting, we are not giving it here.

\subsection{Random Shell}
In the preceding examples, we have studied the kinematics of the probabilistic particle in classical spacetimes. Now we will take advantage of the previous inspections and generalize the spacetime around the spherical thin massive shell to a probabilistic case. The p.st. which we want to introduce should describe the spherical geometry around an \emph{indefinite} thin-shell source, meaning that instead of a fixed value of the shell radius $ \varrho $, we will consider a random variable described by a probability density function $ f(\varrho) $, satisfying
\begin{equation}
\int_{2M}^{\infty} f(\varrho) d \varrho = 1.
\end{equation}
To rule out the possibility of appearance of the horizon (as before), we assume that for $ \varrho \leq 2M $, $ f(\varrho) = 0 $.\\

Our p.st. will be defined as a composite of well-known classical spacetimes. To simplify the problem, we will not consider the 1-parametric composite of spherical shell spacetimes (each of which would be composed of the inside and the outside region) parametrized by $ \varrho $, but rather a mere 2-component discrete composite of Minkowski (coordinates $ T,R $) and Schwarzschild (coordinates $ t,r $) spacetimes, letting their probabilities to be functions of $ \varrho $. This is possible, since the metric in any point is either that of Minkowski or that of Schwarzschild anyway. Now, to define a composite, we have to provide a mapping of the respective coordinates in which the metrics are given (see the end of paragraph 2.1). We choose a mapping\footnote{A physical interpretation of such a mapping is yet missing.} given by $ t = T $ and $ r = R $, which is consistent, since it grants that the metric of any realization is continuous on the shell, no matter its radius. To complete the set-up of the p.st., we provide the probabilities $ P_{\text{S}} $ and $ P_{\text{M}} $ of the Schwarzschild and the Minkowski component, respectively. A slightly nontrivial aspect is that they are functions of the radius:
\begin{equation}
P_{\text{S}}(r) = \int_{2M}^{r} f(\varrho) d \varrho, \qquad P_{\text{M}} = 1 - P_{\text{S}}.
\end{equation}
The p.st. is now defined which automatically gives definition to the integral $ \int dG $. In case of a discrete composite, it is a mere sum over realizations weighted by their probabilities. For a function $ h(g) $ of a metric field realization,
\begin{equation}
\int h(g) dG = h(g_{\text{M}}) P_{\text{M}} + h(g_{\text{S}}) P_{\text{S}}. \label{int_dG_RS}
\end{equation}

Let us establish the description of the probabilistic particle. As before, it shall be spherically symmetric and infalling, leaving the coordinate representation of the p.p.d. in the form $ A(t,r,U_r) $. Here, the vector $ U $ can be chosen equal to the velocity in either the Minkowski or the Schwarzschild realization; we shall make that choice in a moment.\\

The kinematical law for the probabilistic particle in a p.st. is given by (\ref{pp_pst_kinematics}). Plugging in the geometrical quantities computed in preceding paragraphs and using (\ref{int_dG_RS}), we rewrite the equation in the form
\begin{equation}
\begin{aligned}
& \Big[ A 2 r \sin \theta U_r + \left( A,_t U_t + A,_r U_r \right) r^2 \sin \theta \\ & - A,_{U_r} r^2 \sin \theta ~ \frac{M}{r^{2}} \left( \left( 1 - \tfrac{2 M}{r} \right) U_t^{2} - \left( 1 - \tfrac{2 M}{r} \right)^{-1} U_r^{2} \right) \Big] P_{\text{S}}\\
+ & \Big[ A 2 r \sin \theta U_r + \left( A,_t U_t + A,_r U_r \right) r^2 \sin \theta \Big] P_{\text{M}} = 0. \label{RSS_kinematics_orig}
\end{aligned}
\end{equation}
In the coordinate singularities $ \sin \theta = 0 $, the equation holds. In all other cases, we divide by this factor. As we already know, if $ U $ is normalized according to (\ref{u_norm}), it holds
\begin{equation}
\left( \left( 1 - \tfrac{2 M}{r} \right) U_t^{2} - \left( 1 - \tfrac{2 M}{r} \right)^{-1} U_r^{2} \right) = 1.
\end{equation}
Thus, we can simplify considerably by fixing $ U $ to be the velocity in the \break Schwarzschild realization. Also, thanks to $ P_{\text{M}} + P_{\text{S}} = 1 $, we can add up the identical terms, getting
\begin{equation}
A 2 r U_r + \left( A,_t U_t + A,_r U_r \right) r^2 - A,_{U_r} M P_{\text{S}}(r) = 0. \label{RSS_kinematics_half}
\end{equation}
This equation is almost as in the Schwarzschild spacetime, cf. (\ref{Schw_kinematics_simple}), the only difference being the presence of the probability $ P_{\text{S}}(r) $.\\

What we can hope for is to find a static solution, just as before. And just as before, we first have to allow the probabilistic particle to pass the origin. We employ the analytic continuation to include $ r < 0 $ and alter the equation to read
\begin{equation}
A 2 r U_r + \left( A,_t U_t + A,_r U_r \right) r^2 - \text{sgn}(r) A,_{U_r} M P_{\text{S}}(\vert r \vert) = 0. \label{RSS_kinematics}
\end{equation}

Assuming $ A,_t = 0 $, we begin the search for the static solution with
\begin{equation}
A 2 r U_r + A,_r U_r r^2 - \text{sgn}(r) A,_{U_r} M P_{\text{S}}(\vert r \vert) = 0.
\end{equation}
One finds that the solution is of the form
\begin{equation}
A = \frac{\phi(x)}{r^{2}}, \qquad x = U_r^{2} + 2 M \int \frac{\text{sgn}(r) P_{\text{S}}(\vert r \vert)}{r^{2}} dr \label{A_static_RSS}
\end{equation}
where $ \phi $ is an arbitrary function. In the latter, instead of the integral in (\ref{A_static_RSS}), we shall write $ I(r) $.\\

To proceed with the example, let us take an illustrative p.d.f. for the shell radius,
\begin{equation}
f(\varrho) = \frac{\chi_{[\varrho_{0},\varrho_{1}]}(\varrho)}{(\varrho_{1} - \varrho_{0})}, \label{random_shell_uniform}
\end{equation}
where $ 2M < \varrho_{0} < \varrho_{1} $. Then we have
\begin{equation}
I(r) = \begin{cases}
 -\frac{\frac{ \varrho_{0} }{ \varrho_{1} }+\ln
   \left(\frac{ \varrho_{1} }{ \varrho_{0} }\right)-1}{ \varrho_{1} - \varrho_{0} }-\frac{
   1}{ \varrho_{1} } & \left| r\right| < \varrho_{0}  \\
 \frac{\frac{ \varrho_{0} }{\left| r\right| }+\ln \left(\frac{\left| r\right|
   }{ \varrho_{0} }\right)}{ \varrho_{1} - \varrho_{0} }-\frac{\frac{ \varrho_{0} }{\varrho_{1}}+\ln
   \left(\frac{ \varrho_{1} }{ \varrho_{0} }\right)}{ \varrho_{1} - \varrho_{0} }-\frac{1}
   { \varrho_{1} } & \varrho_{0} \leq \vert r \vert \leq \varrho_{1}  \\
 -\frac{1}{\left| r\right| } & \left| r\right| > \varrho_{1} 
\end{cases}
\end{equation}
The integration constant in the region $ \left| r\right| > \varrho_{1}  $ has been set to zero (which does not matter, since this term enters as an argument to an arbitrary function), in the other two regions, it has been set so that the function as a whole is continuous.\\

Now we can consider the previously used form of the solution (\ref{phi_RS}) and again assume the simplest case $ D = const. $ So far, our p.p.d. is
\begin{equation}
A = \frac{D}{r^{2}} ~ \chi_{[-d,d]} \left( U_r^{2} + 2 M I(r) + \tfrac{2M}{r_{\mathrm{in}}} \right).
\end{equation}
Setting $ D = \tilde{D}/(2d) $ and performing the limit $ d \to 0 $ as before, we again reach the sharp solution
\begin{equation}
A = \frac{\tilde{D}}{r^{2}} ~ \delta \left( U_r^{2} + 2 M I(r) + \tfrac{2M}{r_{\mathrm{in}}} \right).
\end{equation}

This is a limiting static solution for the probabilistic particle, directly analogous to the mean geodesic. We plot this solution together with the solutions for thin shells at $ \varrho_{0} $ and $ \varrho_{1} $, as obtained in the preceding paragraph, in Fig.~\ref{fig:random_shell_sharp}. The plot illustrates the averaging and smoothing behaviour of the kinematical equation used.\\

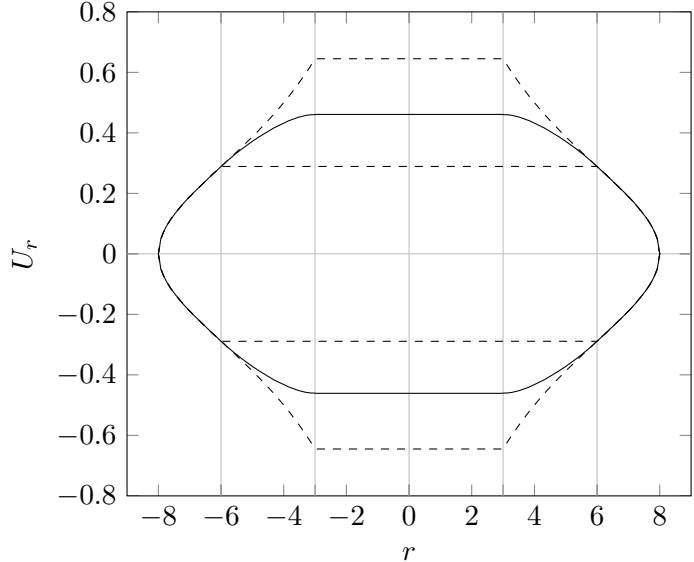
\begin{figure}[ht]
\centering
\begin{tikzpicture}
  \begin{axis}[
	width=9 cm,
	height=8 cm,
    xmin=-9,xmax=9,
    ymin=-0.8,ymax=0.8,
    xlabel=$ r $,
    ylabel={$ U_r $},
    ylabel near ticks,
    xlabel near ticks,
    	/pgf/number format/.cd,
    	1000 sep={},
    extra y ticks       = 0,
    extra y tick labels = ,
    extra y tick style  = { grid = major },
    extra x ticks       = {-6,-3,0,3,6},
    extra x tick labels = ,
    extra x tick style  = { grid = major },
    ]
  
    
  \addplot[black,dashed] table {plots/RSSdelta1.txt};
  \addplot[black,dashed] table {plots/RSSdelta1neg.txt};
  
   \addplot[black,dashed] table {plots/RSSdelta2.txt};
  \addplot[black,dashed] table {plots/RSSdelta2neg.txt};
  
   \addplot[black] table {plots/RSSdelta3.txt};
  \addplot[black] table {plots/RSSdelta3neg.txt};
 
  \end{axis}

\end{tikzpicture}
\vspace{-2 mm}
\caption{The $ \delta $-limit of the p.p.d. plotted as $ U_r(r) $ (full line), in comparison with the analogical solutions to a thin shell (see Fig.~\ref{fig:shell}) at $ \varrho = \varrho_0 $ and $ \varrho = \varrho_1 $ (the inner velocities have been rescaled again). The setting is $ M = 1 $, $ \varrho_0 = 3 $, $ \varrho_1 = 6 $ and $ r_{\mathrm{in}} = 8 $.}
\label{fig:random_shell_sharp}
\end{figure}

The probabilistic particle in the vicinity of a random shell, driven by (\ref{pp_pst_kinematics}), behaves as if it was in the Schwarzschild spacetime with a rigged mass parameter: substituting $ \tilde{M} = M P_S(r) $ in (\ref{RSS_kinematics_half}) produces the equation (\ref{Schw_kinematics_simple}) with $ \tilde{M}(r) $ instead of $ M $. Our result is showing good agreement of the probabilistic particle formalism with the conventional test particle, yet does not go beyond that.\\

Our last effort towards depicting the effects of the random shell on a test object uses the prescription (\ref{A1_unexpanded}), artificially choosing time slices separated by a small $ \delta t $, on which the two versions of the p.p.d., evolved separately along respective metrics of the p.st., are summed up again\footnote{This process is a distant parallel to the collapse of wavefunction in quantum mechanics, which may be a partial justification for choosing the hypersurface where the collapse happens, i.e., the hypersurface of ``the measurement''.}.\\

We shall choose an initial condition of the Gaussian form
\begin{equation}
A(t_0,r,U_r) = N \exp \left(-\frac{(U_r - \bar{U}_r)^2}{\varepsilon_U} \right) \exp \left( -\frac{(r - \bar{r})^2}{\varepsilon_r} \right) \label{initialA}
\end{equation}
and proceed with the step-by-step evolution numerically. We will take small time steps $ \delta t $ and use a linear approximation of the argument, as in (\ref{A1}) (we can do this as long as $ U_t $ is not too small). We are interested in the change of shape of the p.p.d. when subjected to probabilistic geometry. The resulting step-by step evolution for a particle infalling through the random shell given by (\ref{random_shell_uniform}) is plotted in Fig.~\ref{fig:middle_incoming_r} and Fig.~\ref{fig:middle_incoming_U}. In the course of the calculation, we have set the volume element to 1 (instead of $ \sqrt{\vert g_{M,S} \vert} = r^{2} \sin \theta $) to make visible the effect of the probabilistic geometry (otherwise the p.p.d. would additionally decrease as $ \sim r^{-2} $). Also, we are putting aside the question of normalization.\\

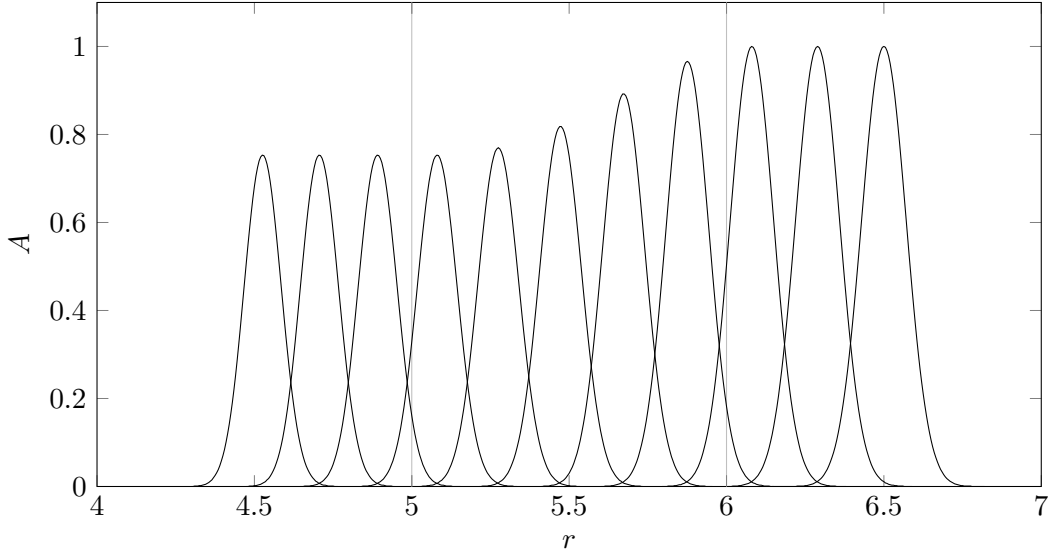
\begin{figure}[H]
\centering
\begin{tikzpicture}
  \begin{axis}[
	width=14 cm,
	height=8 cm,
    xmin=4,xmax=7,
    ymin=0,ymax=1.1,
    xlabel=$ r $,
    ylabel={$ A $},
    ylabel near ticks,
    xlabel near ticks,
    	/pgf/number format/.cd,
    	1000 sep={},
    extra x ticks       = {5,6},
    extra x tick labels = ,
    extra x tick style  = { grid = major },
    xtick={4, 4.5, 5, 5.5, 6, 6.5, 7},
    cycle list name=color list,
    ]
    
  \addplot[black] table {plots/RSS_middle_incoming_r0.txt};
  \addplot[black] table {plots/RSS_middle_incoming_r1.txt};
  \addplot[black] table {plots/RSS_middle_incoming_r2.txt};
  \addplot[black] table {plots/RSS_middle_incoming_r3.txt};
  \addplot[black] table {plots/RSS_middle_incoming_r4.txt};
  \addplot[black] table {plots/RSS_middle_incoming_r5.txt};
  \addplot[black] table {plots/RSS_middle_incoming_r6.txt};
  \addplot[black] table {plots/RSS_middle_incoming_r7.txt};
  \addplot[black] table {plots/RSS_middle_incoming_r8.txt};
  \addplot[black] table {plots/RSS_middle_incoming_r9.txt};
  \addplot[black] table {plots/RSS_middle_incoming_r10.txt};
  \end{axis}

\end{tikzpicture}
\vspace{-2 mm}
\caption{Sections of $ A $ in the maximum w.r.t. $ U_r $. We plot the function as evolved by the step-by-step prescription, infalling through a random shell. The setting is $ M = 1 $, $ \varrho_0 = 5 $, $ \varrho_1 = 6 $ and $ \delta t = 0.4 $. The initial condition is (\ref{initialA}) with $ \bar{r} = 6.5 $, $ \varepsilon_r = 0.01 $, $ \bar{U}_r = -1 $, $ \varepsilon_U = 0.0001 $, and for simplicity $ N = 1 $.}
\label{fig:middle_incoming_r}
\end{figure}

\begin{figure}[ht]
\centering
\begin{tikzpicture}
  \begin{axis}[
	width=14 cm,
	height=8 cm,
    xmin=-1.08,xmax=-0.95,
    ymin=0,ymax=1.1,
    xlabel=$ U_r $,
    ylabel={$ A $},
    ylabel near ticks,
    xlabel near ticks,
    	/pgf/number format/.cd,
    	1000 sep={},
    xtick={-1.06, -1.04, -1.02, -1, -0.98, -0.96},
    cycle list name=color list,
    ]
    
  \addplot[black] table {plots/RSS_middle_incoming_U0.txt};
  \addplot[black] table {plots/RSS_middle_incoming_U1.txt};
  \addplot[black] table {plots/RSS_middle_incoming_U2.txt};
  \addplot[black] table {plots/RSS_middle_incoming_U3.txt};
  \addplot[black] table {plots/RSS_middle_incoming_U4.txt};
  \addplot[black] table {plots/RSS_middle_incoming_U5.txt};
  \addplot[black] table {plots/RSS_middle_incoming_U6.txt};
  \addplot[black] table {plots/RSS_middle_incoming_U7.txt};
  \addplot[black] table {plots/RSS_middle_incoming_U8.txt};
  \addplot[black] table {plots/RSS_middle_incoming_U9.txt};
  \addplot[black] table {plots/RSS_middle_incoming_U10.txt};
 
  \end{axis}

\end{tikzpicture}
\vspace{-2 mm}
\caption{Sections of $ A $ in the maximum w.r.t. $ r $. The setting is that of Fig.~\ref{fig:middle_incoming_r}.}
\label{fig:middle_incoming_U}
\end{figure}
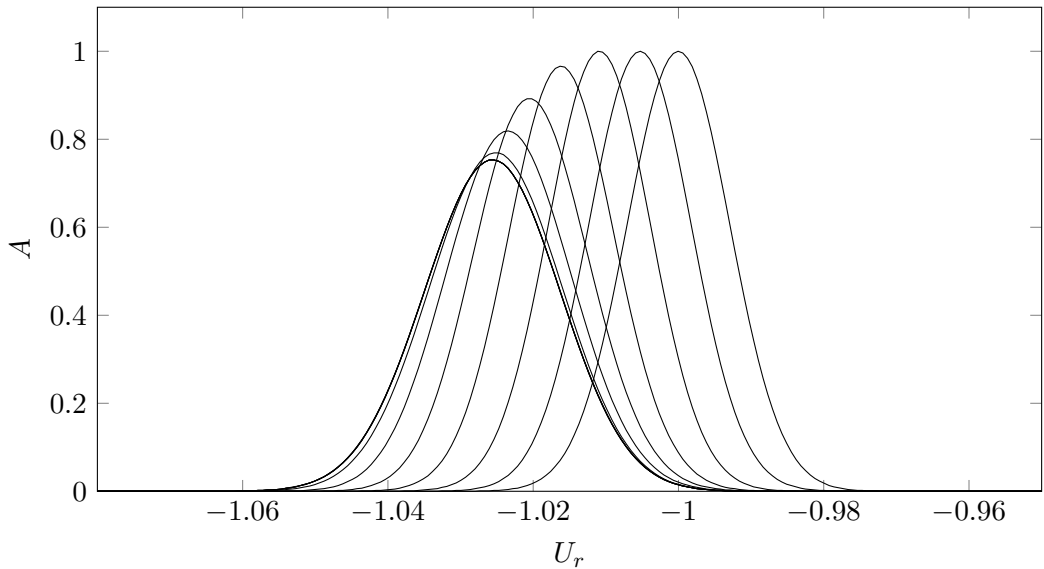

In the first plot, we can see that the shape of the section is not much changing until it reaches the probabilistic region $ \varrho_{0} \leq r \leq \varrho_{1} $, where it undergoes mixing of the Schwarzschild and the Minkowski component in each step, and dissolves accordingly. Once it gets to the interior of the shell, the shape of the section is again stabilized.\\

In the second plot, we can see what happens with velocity. First, the particle is accelerated by the Schwarzschild component, and the velocity distribution is shifting towards higher (absolute) values. In the probabilistic region, it dissolves, only to converge to a fixed shape in the flat interior.\\

To give the complete picture, we plot the initial and final p.p.d. in $ r $-$ U_r $ plane in Fig.~\ref{fig:3D}. One can spot that the final p.p.d. is not only dissolved in velocities, but also slightly tilted, which is due to the fact that particles of different velocities evolve in different ways.

\begin{figure}[ht]
\centering
 \begin{tikzpicture}[baseline]
 \begin{axis}[
 	width=6.6 cm,
	height=6.8 cm,
	view={0}{90},
	plot box ratio=1 1 1,
	colormap={backwhite}{[1cm] rgb255(0cm)=(255,255,255) rgb255(1cm)=(40,40,40)},
	shader=interp,
	colorbar/width=4 mm,
	ymin=-1.025,ymax=-0.975,
    xmin=6.2,xmax=6.8,
    zmin=0,zmax=1,
	xlabel={$ r $},
	ylabel={$ U_r $},
	zlabel={},
    ylabel near ticks,
    xlabel near ticks,
		/pgf/number format/.cd,
    	1000 sep={},
]
\addplot3[surf,shader=interp] file {plots/3Dplot_initial.csv};
\end{axis}
\end{tikzpicture}
\hskip 0 pt
\begin{tikzpicture}[baseline]
 \begin{axis}[
 	width=6.6 cm,
	height=6.8 cm,
	view={0}{90},
	plot box ratio=1 1 1,
	colormap={backwhite}{[1cm] rgb255(0cm)=(255,255,255) rgb255(1cm)=(40,40,40)},
	shader=interp,
	colorbar/width=4 mm,
	ymin=-1.05,ymax=-1.00,
    xmin=4.2,xmax=4.8,
    zmin=0,zmax=1,
	xlabel={$ r $},
	ylabel={},
	zlabel={},
    ylabel near ticks,
    xlabel near ticks,
		/pgf/number format/.cd,
    	1000 sep={},
]
\addplot3[surf,shader=interp] file {plots/3Dplot_final.csv};
\end{axis}
\end{tikzpicture}
\caption{An overall plot of the initial condition (on the left) and an overall plot of the p.p.d., as evolved by the 10 successive steps (on the right). The setting is that of preceding figures.}
\label{fig:3D}
\end{figure}
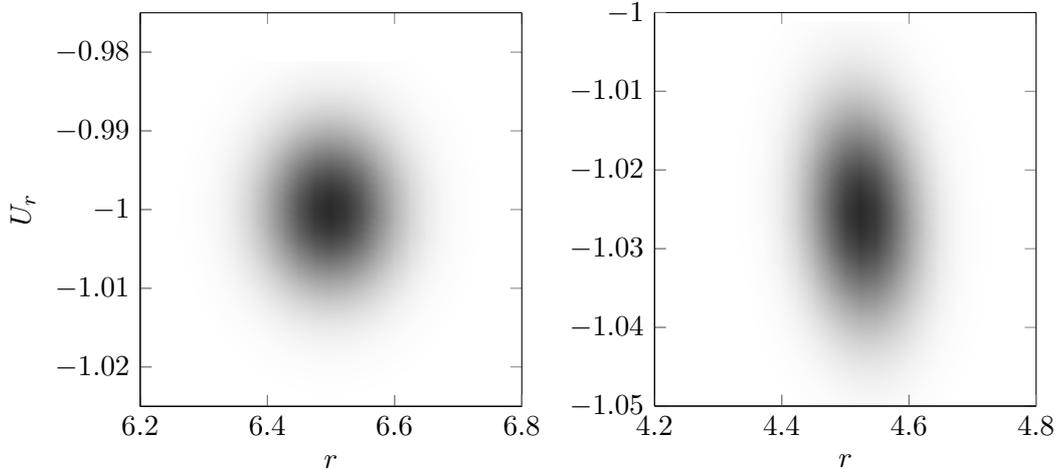

\section*{Conclusion}
\addcontentsline{toc}{section}{Conclusion}

On the foregone pages, we have introduced and studied a concept inspired by certain reflections on (path integral) quantum gravity approaches and the micro-structure of relativistic models of the world. Should we name the major points of the work done, they would be the following:
\begin{enumerate}[label=\roman*., itemsep=-0.2mm]
\item A definition of probabilistic spacetime has been provided. It is based upon the probability theory and implicitly incorporates the rules for computing probabilities of sets of spacetime realizations.
\item An easy technique for constructing rough measures for a general p.st., using so-called fragmentation, has been suggested and related to the notion of locally probabilistic spacetime.
\item A way to construct a special, extremely simple p.st. (the composite) with detailed measure has been provided. It is an implementation of the intuitive idea of collection of spacetimes endowed with probabilities.
\item The kinematics of the test point particle in the p.st. has been discussed, and a description of the probabilistic particle, its generalization, has been provided. Kinematics of the probabilistic particle has then become the subject of detailed consideration.
\item An example of a p.st. (the random shell) has been presented and studied, using the previously introduced formalism.
\end{enumerate}

Thanks to the simplicity and generality of the concept, we continue to see certain application potential in it. It was already mentioned that the probabilistic spacetime may describe the result of a quantum gravity computation, and if it does, it could be of benefit to understand its structure. Moreover, it turns out that the provided definition, intended to be as simple and as natural as possible, rests upon the most intriguing object, which is the space of metric field realizations  $ \mathscr{R} $. In the suggested formalism, defining a probabilistic spacetime goes hand in hand with defining an integral over some subsets of $ \mathscr{R} $. No need to say that integrating over $ \mathscr{R} $ is the main subject of research conducted within all path integral approaches to quantum gravity. It has thus happened that we have naturally returned to the topic which had lead our thoughts in the beginning. In any case, our belief is that the work done in between brings sometimes naive, yet novel (and maybe inspiring) viewpoint on this matter.

\def\bibname{Bibliography}

\end{document}